\documentclass[twocolumn,journal]{IEEEtran}

\usepackage{ifpdf}
\usepackage{url}
\usepackage{amsmath}
\usepackage{subfigure}
\usepackage{comment}
\usepackage{color}
\usepackage{enumerate}
\usepackage{cite}
\usepackage{amssymb}
\usepackage{bbm}
\usepackage{mleftright}
\usepackage{booktabs}
\usepackage{diagbox}
\usepackage{multirow}
\usepackage{hhline}
\usepackage{balance}
\usepackage{epstopdf}

\usepackage{amsthm}
\newtheorem{mydef}{Proposition}

\ifCLASSINFOpdf
  \usepackage[pdftex]{graphicx}
  \DeclareGraphicsExtensions{.pdf,.png,.bmp}
\else
  \usepackage[dvips]{graphicx}
  \DeclareGraphicsExtensions{.eps}
\fi

\title{On the Temporal Effects of Mobile Blockers in Urban Millimeter-Wave Cellular Scenarios}
\author{\IEEEauthorblockN{Margarita Gapeyenko, Andrey Samuylov, Mikhail Gerasimenko, Dmitri Moltchanov, Sarabjot Singh,\\
Mustafa Riza Akdeniz, Ehsan Aryafar, Nageen Himayat, Sergey Andreev, and Yevgeni Koucheryavy
\vspace{-0.8cm}
\thanks{Copyright (c) 2015 IEEE. Personal use of this material is permitted. However, permission to use this material for any other purposes must be obtained from the IEEE by sending a request to pubs-permissions@ieee.org.}
\thanks{M.~Gapeyenko, A.~Samuylov, M.~Gerasimenko, D.~Moltchanov, S.~Andreev, and Y.~Koucheryavy are with Tampere University of Technology, Tampere, Finland
(e-mail: \{firstname.lastname, evgeni.kucheryavy\}@tut.fi).}
\thanks{S.~Singh, M.~R.~Akdeniz, N.~Himayat are with Intel Corporation, Santa Clara, CA, USA (e-mail: sarabjotsingh.in@gmail.com, \{mustafa.akdeniz, nageen.himayat\}@intel.com)}
\thanks{E.~Aryafar is with Portland State University, Portland, OR, USA (e-mail: earyafar@gmail.com)}
\thanks{This work was supported by Intel Corporation and the Academy of Finland. The work of S. Andreev was supported in part by a Postdoctoral Researcher grant from the Academy of Finland and in part by a Jorma Ollila grant from Nokia Foundation.}
}
}

\begin{document}

\maketitle

\begin{abstract}
Millimeter-wave (mmWave) propagation is known to be severely affected by the blockage of the line-of-sight (LoS) path. In contrast to microwave systems, at shorter mmWave wavelengths such blockage can be caused by human bodies, where their mobility within environment makes wireless channel alternate between the blocked and non-blocked LoS states. Following the recent 3GPP requirements on modeling the \textit{dynamic blockage} as well as the \textit{temporal consistency} of the channel at mmWave frequencies, in this paper a new model for predicting the state of a user in the presence of mobile blockers for representative 3GPP scenarios is developed: urban micro cell (UMi) street canyon and park/stadium/square. It is demonstrated that the blockage effects produce an \textit{alternating renewal process} with exponentially distributed non-blocked intervals, and blocked durations that follow the general distribution. The following metrics are derived (i) the mean and the fraction of time spent in blocked/non-blocked state, (ii) the residual blocked/non-blocked time, and (iii) the time-dependent conditional probability of having blockage/no blockage at time $t_{1}$ given that there was blockage/no blockage at time $t_{0}$. The latter is a function of the arrival rate (intensity), width, and height of moving blockers, distance to the mmWave access point (AP), as well as the heights of the AP and the user device. The proposed model can be used for system-level characterization of mmWave cellular communication systems. For example, the optimal height and the maximum coverage radius of the mmWave APs are derived, while satisfying the required mean data rate constraint. The system-level simulations corroborate that the use of the proposed method considerably reduces the \textit{modeling complexity}.
\end{abstract}



\vspace{-0.3cm}
\section{Introduction}\label{sec:intro}


The rapidly growing number of mobile devices as well as the associated growth of mobile traffic call for an unprecedented increase in access capacity. To meet more stringent performance requirements, the use of the so-called cellular millimeter-wave (mmWave) technology operating at frequencies such as $28$~GHz and $73$~GHz has been proposed in fifth-generation (5G) mobile systems~\cite{Andrews2, nokia, Rappaport_13_mmWave}.


Together with the phenomenal increase in access capacity, the use of the extremely high frequency (EHF) bands creates unique challenges for wireless communication systems. One of them is a need for development of appropriate mmWave channel models. Indeed, various groups and organizations have recently developed a number of such channel models~\cite{metis, standard_16, cost, mmW_Ch_M_16_JSAC_Rapaport, 5GPL_rapaport, Samimi_16, Sun_ICC17}. In contrast to microwave systems, the propagation characteristics of mmWave systems (with wavelengths of under a centimeter) are impacted not only by larger objects such as buildings, but also by much smaller obstacles such as cars, lampposts, and even humans. Given that mmWave systems are envisioned to be deployed in urban squares and streets, 3GPP has identified humans as one of the major factors affecting the mmWave propagation and has incorporated a blockage model into TR 38.901 of Release 14~\cite{standard_16}.


The performance of mobile communications systems is typically characterized by developing system-level simulation (SLS) frameworks~\cite{Zhang_INFOCOM_16, Mezzavilla_ACN_15}. Modeling the path loss with simple power-law abstractions, these SLS tools may take into account the necessary details of the target technologies and deliver their output results within a reasonable time. However, when conducting system-level evaluation of a mmWave system, in addition to the path loss model that captures the propagation environment, one needs to explicitly represent and track all of the relevant static and mobile objects with dimensions larger than a few centimeters. This significantly increases the required computational resources and expands simulation time.

Motivated by the new effects in mmWave communications systems as well as by the recent 3GPP requirements for 5G channel modeling, this paper studies the \textit{dynamic blockage} caused by humans in outdoor urban mmWave cellular deployments, while specifically concentrating on the \textit{temporal consistency} of the link states for a static user.

\vspace{-0.2cm}
\subsection{Background and Related Work}\label{sec:related}

The importance of dynamic blockage of the LoS path in mmWave deployments has recently been shown to be one of the critical design factors that affect system performance~\cite{metis, standard_16, cost, Petrov_JSAC17}. An example illustration of the measured path gain experienced by a node in a realistic crowded environment is shown in Fig.~\ref{fig:pathGain}. As one may observe, dynamic blockage by small mobile objects within the environment, such as moving people, cars, trucks, etc., introduces additional uncertainty in the channel, which may eventually result in sharp drops (up to $30\sim40$~dB) in the received signal strength \cite{dynamicBlockage, Weiler_WCL_16}. The blockage frequency, duration, and the resultant degradation of signal strength affect the performance of a mmWave system.

\begin{figure}[!t]
\centering
\includegraphics[width=0.8\columnwidth]{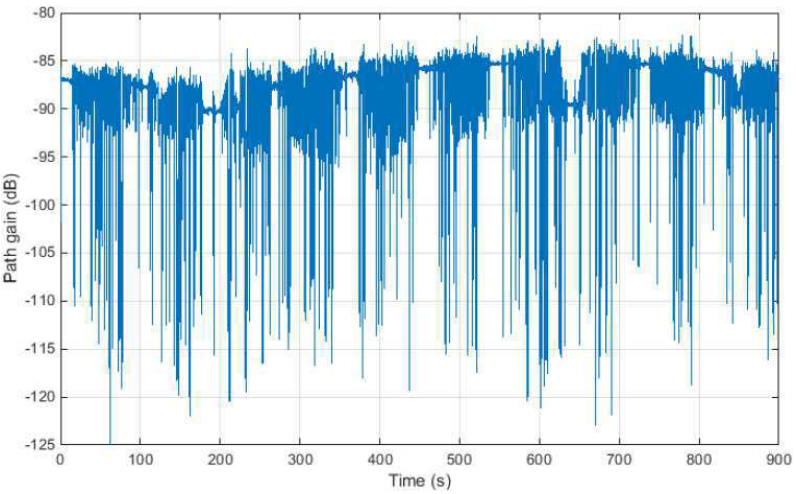}
\caption{Path gain in presence of dynamic blockage, reproduced from \cite{dynamicBlockage}.}
\vspace{-0.5cm}
\label{fig:pathGain}
\end{figure}


Recent work has studied the impact of LoS blockage in urban microwave systems~\cite{heath, heath2}. However, the results do not directly apply to mmWave systems as the objects of interest in mmWave and microwave systems are of fundamentally different nature and hence would require different models for their accurate representation. Indeed, in addition to mobility of smaller obstacles, such as humans, one also needs to take into account their inherently random dimensions.


The LoS blockage by humans in mmWave systems has been evaluated through simulation studies in~\cite{charlton}. In~\cite{gapeyenkoICC}, a LoS blockage model where humans are represented as cylinders of random width and height was proposed. However, there the authors assumed that both the users and the blockers are stationary. In addition to academic work, the 3GPP community is currently exploring various options for modeling the impact of human blockage appropriately~\cite{dynamicBlockage,channel_model_white}.

In~\cite{standard_16}, the human body blockage is taken into account by creating rectangular screens dropped onto the simulation map. A similar approach is adopted by~\cite{jacob_11_EUCAP, knife_edge}, where the authors also evaluated the accuracy of their methods. Due to the properties of the propagation model, which generates a random sample of the propagation path at each run, a particular attention of the 3GPP work groups is being paid to spatial and temporal consistency of the mmWave links~\cite{standard_16,channel_model_white}.

In~\cite{Dettmann16_Tem_Cor}, the authors contributed a model for temporal correlation of interference in a mobile network with a certain density of users. It was demonstrated that the correlated propagation states across the users significantly impact the temporal interference statistics. Analytically tractable models for correlated outdoor and indoor shadowing have been proposed in~\cite{heath_correlation_shadowing} and \cite{correlation_shadowing_outdoor}, thus accentuating the high correlation between the locations of the nodes and the shadowing effects. The analytical expression to characterize the correlation between the signals of two antennas was given in~\cite{Yacoub_corr}. 


Even though there has been a considerable literature coverage on user mobility in general~\cite{mobility_TVT, real_mobility_16_VTC, Jacob_10,Kashiwagi_TVT_10}, to the best of our knowledge there are only a few studies that incorporate the user mobility into analytically tractable models~\cite{Samuylov_GCW16, Madadi_WiOpt16}. These latest results confirm the presence of memory in the LoS blockage process and highlight its dependence on the mobility characteristics of the users.

\vspace{-0.2cm}
\subsection{State-of-the-Art and Contributions}


The goal of this paper is to contribute a novel mathematical methodology that aims to characterize the dynamics and the temporal correlation of LoS human body blockage statistics. In this work, a model of the LoS blockage for a stationary user in a moving field of blockers is proposed. This scenario is more typical for outdoor mmWave systems as compared to stationary blockage models assumed in prior work. The blockers are modeled as cylinders of a certain height and width that enter the LoS zone of a mmWave receiver according to a Poisson process in time.

The analysis is based on the combined application of stochastic geometry, renewal process theory, and queuing models. Three different scenarios are addressed, including two street canyon use cases and a park layout (see Fig.~\ref{fig:analytic}). The metrics of interest are those reflecting temporal behavior of the LoS blockage process, such as the mean and the fraction of blocked/non-blocked LoS, the residual time in blocked/non-blocked states, and the time-dependent effects of conditional blocked/non-blocked state probabilities. 

\begin{figure*}[!t]
\centering
\subfigure[\label{fig:analytic_a}{Sidewalk 1, S1}]
{\includegraphics[width=.28\textwidth]{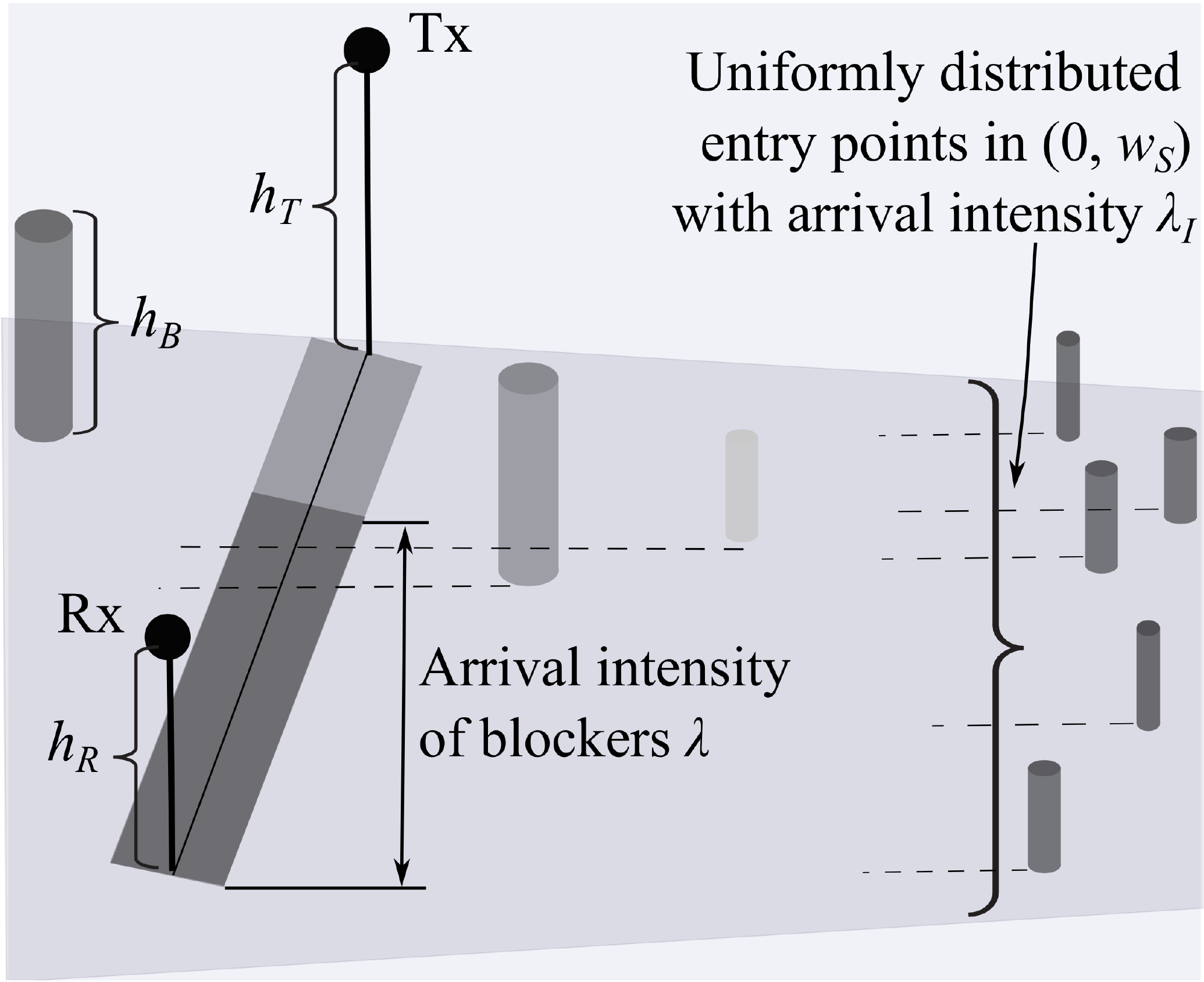}}~~~~~~~
\subfigure[\label{fig:analytic_b}{Sidewalk 2, S2}]
{\includegraphics[width=.28\textwidth]{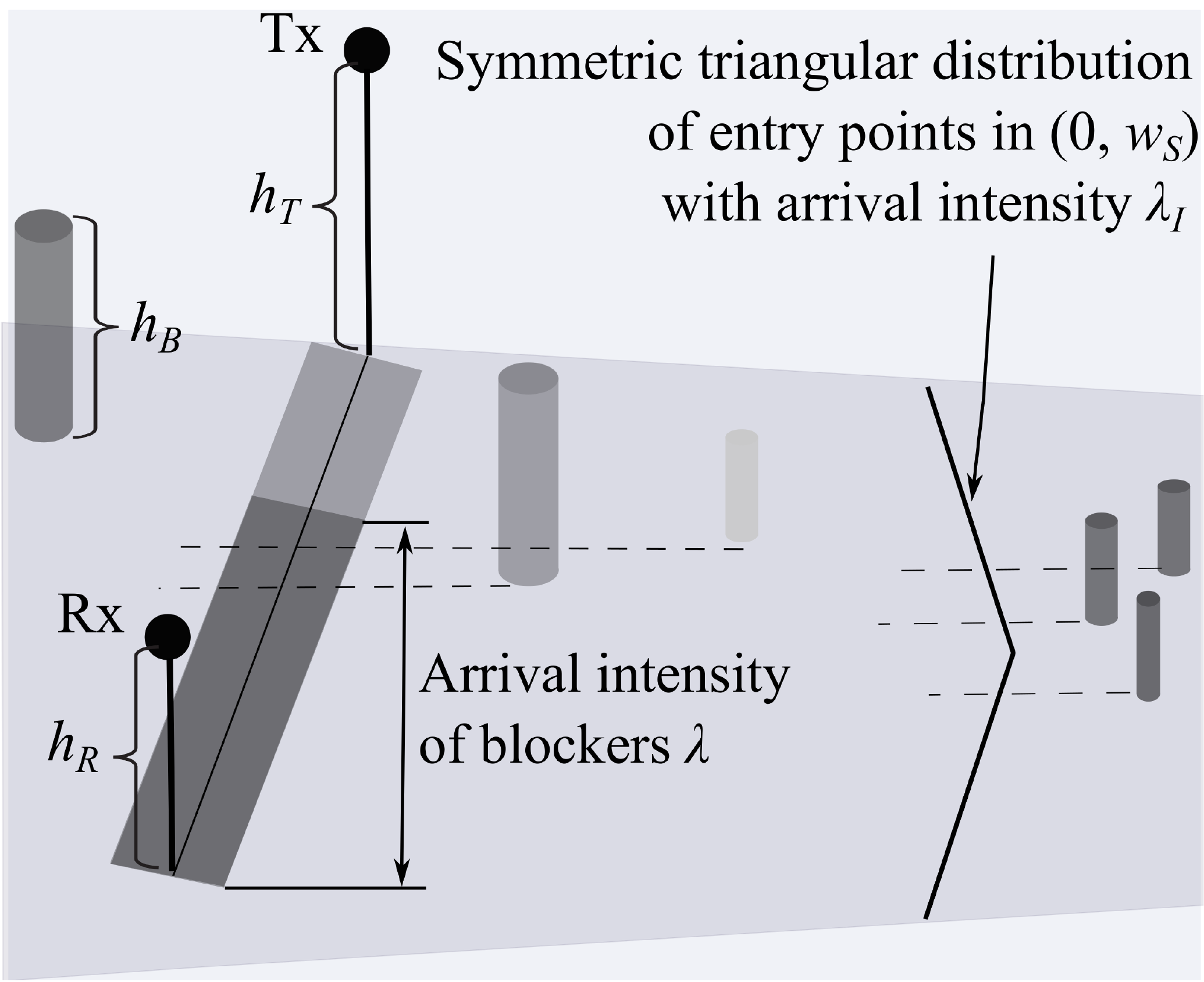}}~~~~~~~
\subfigure[\label{fig:analytic_c}{Park/stadium/square, S3}]
{\includegraphics[width=.28\textwidth]{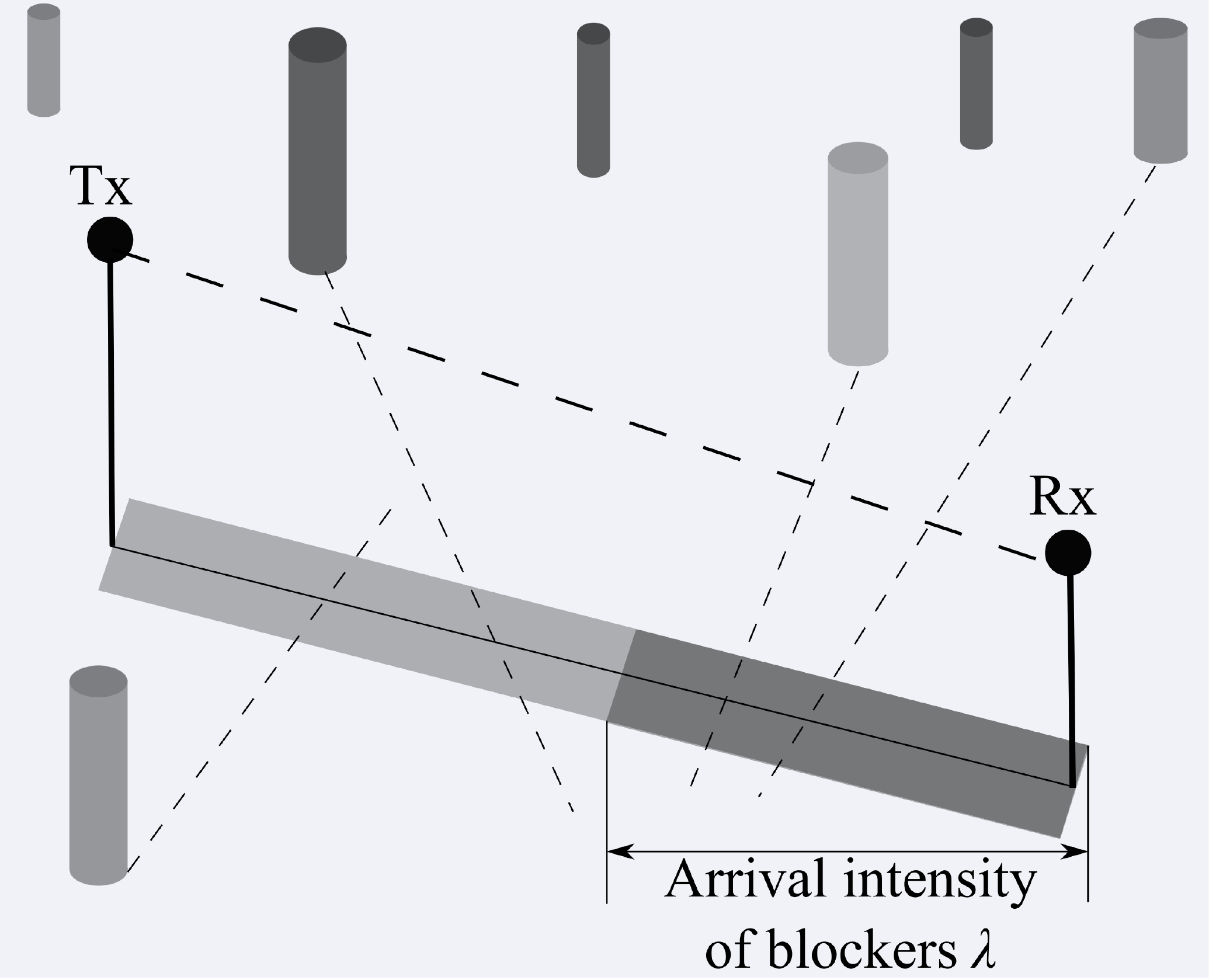}}
\caption{Three considered scenarios for further analytical modeling.}
\label{fig:analytic}
\end{figure*}


In summary, the following contributions are delivered by this work:
\begin{itemize}
	\item{To analyze the \textit{temporal correlation} and the \textit{dynamic blockage process} by human bodies at mmWave frequencies, a novel mathematical model is proposed. It is shown that the analytical expression could be utilized to replace explicit simulation of the mobile blockers in the SLS studies. The associated improvement in the simulation times depends on the crowd intensity and may reach several orders of magnitude.}
  \item{To capture the general structure of the dynamic LoS blockage process, including the impact of mobile obstacles, the corresponding mathematical methodology is developed. It is observed that non-blocked/blocked periods form an \textit{alternating renewal process} where the non-blocked intervals follow an exponential distribution and the blocked intervals have a general distribution. The latter is captured by employing methods for the busy period analysis in the $M/GI/\infty$ queuing model.}
  \item{To characterize the \textit{temporally consistent} human body blockage process, a simplified approach is developed to calculate the \textit{conditional probabilities}. It is demonstrated that for realistic input parameter values, in all the considered scenarios there always is a significant dependence between the states of the user at $t_{0}$ and $t_{1}$ over small timescales.}
  \item{To demonstrate the applicability of the proposed methodology, the optimal height of the mmWave AP that maximizes the average time in non-blocked LoS conditions as well as the maximum coverage radius that satisfies the required mean data rate are estimated.}
\end{itemize}


The rest of this paper is organized as follows. In Section~\ref{sec:system}, the system model and a description of the outdoor scenarios proposed by 3GPP, which reflect real-life mmWave system usage situations, are introduced. The analysis for the performance metrics of interest is summarized in Section~\ref{sec:model}. The numerical results, particularly those related to the temporal dependencies in the LoS blockage process, are discussed in Section~\ref{sec:results}. Section~\ref{sect:apps} elaborates on the applications of the proposed methodology. Conclusions are presented in the last section of the paper.

\section{System Model}\label{sec:system}

\subsection{General Considerations}

The proposed system model is illustrated in Fig.~\ref{fig:analytic}. The transmitter (Tx) and the receiver (Rx) are deployed at the heights of $h_T$ and $h_R$ from the ground, respectively. The two-dimensional distance between the Tx and Rx is $r_0$. Following~\cite{cylinders}, the potential blockers (i.e., humans) are modeled as cylinders with the height of $h_{B}$ and the base diameter of $d_{m}$.

Note that there always is an area between Tx and Rx, where the emergence of a blocker will cause occlusion of the mmWave LoS link. With the above parameters, this area may be approximated by a rectangular shape, named here the \textit{LoS blockage zone} and denoted as $ABCD$ in Fig.~\ref{fig:geometrical_scenario}. The particular dimensions of this area, its geometrical shape, and the position with respect to Tx and Rx can be estimated given the aforementioned parameters as discussed in what follows.

\begin{table}[b!]\footnotesize
\centering
\caption{Summary of notation and parameters}
\vspace{-0.2cm}
\begin{tabular}{p{2.1cm}p{5.9cm}}
\hline
\textbf{Notation}&\textbf{Description} \\
\hline
$h_{T}$, $h_{R}$, $h_B$ & Height of Tx, Rx, blockers\\
$r_0$ & Two-dimensional distance between Tx and Rx\\
$d_m$, $V$ & Diameter and speed of blockers\\
$w_S$ & Width of the sidewalk\\
$w_E$, $r$  & Effective width and length of LoS blockage zone\\
$\lambda_{I}$ & Initial arrival intensity of blockers per time unit\\
$\lambda$ & Arrival intensity of blockers entering the LoS blockage zone per time unit\\
$\lambda_{S}$ & Arrival intensity of blockers entering the unit area of LoS blockage zone\\
$\lambda_{N}$ & Density of users per unit area\\
$y_A$, $y_B$, $y_C$, $y_D$ & y-coordinates of the edges of the LoS blockage zone\\
$\alpha$ & Angle between Y-axis and the segment Tx-Rx\\
$L$ & Distance walked by a blocker in LoS blockage zone\\
$T$ & Residence time of a blocker in LoS blockage zone\\
$\omega_{j}$, $\eta_{j}$ & The non-blocked and blocked time interval\\
$F_{\omega}(x)$, $E[\omega]$ & CDF, the mean of non-blocked time interval\\
$F_{\eta}(x)$, $E[\eta]$  & CDF, the mean of blocked time interval\\
$F_{T}(x), f_{T}(x), \mathbb{E}[T]$& CDF, pdf, the mean of LoS zone residence time\\
$F_{Y}(x)$  & CDF of the y-coordinate of blocker entry point\\
$F_{\widetilde{Y}}(x)$ & Truncated distribution of the entry point defined on $y_A \leq x \leq y_C$\\
$F_{L}$, $f_{L}$& CDF and pdf of the residence distance $L$\\
$\mathbb{E}[T_l]$, $\mathbb{E}[T_n]$ & Fraction of time in non-blocked/blocked states\\
$F_{t_{\omega}}(x)$, $F_{t_{\eta}}(x)$ & Residual time in non-blocked/blocked states\\
$\xi_{j}$ & $j^{th}$ time interval equal to $\omega_{j} + \eta_{j}$\\
$F_{\xi}(x), f_{\xi}(x), E[\xi]$ & CDF, pdf, the mean of $\omega_{j} + \eta_{j}$\\
$f(x)$ & pdf of renewal process\\
$p_{00}$, $p_{01}$ & Conditional probabilities to be in non-blocked/blocked states at time $t_1$ ($0$ and $1$) given that there was non-blocked state at $t_0$\\
$p_{10}$, $p_{11}$ & Conditional probabilities to be in non-blocked/blocked states at time $t_1$ ($0$ and $1$) given that there was blocked state at $t_0$\\
\hline
\end{tabular}
\label{tab:analysis_parameters}
\end{table}

The speed of blockers $V$ is assumed to be constant. However, the actual mobility model of blockers depends on the scenario as introduced below. The main parameters and the description of employed notation are collected in Table~\ref{tab:analysis_parameters}.

\vspace{-0.3cm}
\subsection{Blocker Mobility and Arrival Modeling}\label{sec:mob_models}


To characterize the human mobility, the following three scenarios are considered:
\begin{itemize}
  \item{\textit{Sidewalk 1 (First scenario, S1).} In this scenario, the mmWave Tx (the AP) is assumed to be mounted on the wall of a building while the Rx may reside at any location on the sidewalk of width $w_S$ within the coverage area of the mmWave AP. The blockers move along the straight line parallel to each other and the side of the sidewalk at a constant speed of $V$ while their y-coordinates of crossing the width of the sidewalk are distributed uniformly within $(0,w_S)$, see Fig.~\ref{fig:analytic}(a). The arrival process of blockers crossing a vertical line -- the width of the sidewalk $w_S$ -- is Poisson in time with the arrival intensity $\lambda_I$.}
  \item{\textit{Sidewalk 2 (Second scenario, S2).} This scenario is similar to the previous case, except for how the blocker positions are distributed in the sidewalk. In practice, the users tend to move closer to the center of the walkway. Therefore, y-coordinates of crossing the width of the sidewalk are modeled by employing a symmetric triangular distribution over $(0,w_S)$, see Fig.~\ref{fig:analytic}(b). The arrival process of blockers crossing the width of the sidewalk is again Poisson in time with the arrival intensity $\lambda_I$ of blockers per time unit.}
  \item{\textit{Park/Stadium/Square (Third scenario, S3).} In this scenario, the users are allowed to enter and leave the mmWave LoS blockage zone at any point along the three sides of the rectangle, see Fig.~\ref{fig:analytic}(c). It is assumed that both the entry and the exit points are distributed uniformly over the side lengths for each individual user. The arrival process of users into the \textit{LoS blockage zone} is Poisson with the arrival intensity of $\lambda_I$ per time unit.}
\end{itemize}


The proposed methodology generally allows to capture more specific types of blocker mobility. For example, one may decide to relax the assumption of the straight movement and thus model the \textit{walking street} environment, where the user trajectories are not required to remain parallel to the sides of the street. Also note that the straight trajectories of blocker mobility inside the LoS blockage zone are the direct consequence of small dimensions of the said zone, hence resulting in negligible changes of behavior with respect to the angle of motion.


The considered metrics of interest are those pertaining to the temporal behavior of the LoS blockage process and include (i) the mean and the fraction of time in the blocked/non-blocked state as experienced by the Rx, (ii) the residual time in the blocked/non-blocked state, and (iii) the conditional probability that there is blocked/non-blocked state at $t_{1}$ given that there was blocked/non-blocked state at $t_{0}$, $t_{1}>t_{0}$.


\section{Proposed System Analysis}\label{sec:model}


All of the three scenarios introduced in the previous Section~\ref{sec:system} can be characterized by following the proposed methodology. The key difference between them is in the distribution of the residence time in the LoS blockage zone (that is, the time that a blocker spends in the LoS blockage zone while crossing it). In this section, the general method to obtain the distribution of the zone residence time in the LoS blockage for the first scenario (i.e., Sidewalk 1, S1) is described. For the second and the third scenarios, the corresponding derivations are reported in Appendix~\ref{appendix:appendix_A}. Finally, the target metrics of interest are produced.


The step-by-step analytical approach may be summarized as follows:
\begin{itemize}
  \item{Specify the zone where blockers may occlude the LoS path and thus determine the LoS blockage zone geometry;}
  \item{Describe the process of blockage by introducing the alternating renewal process that captures the non-blocked/blocked intervals;}
  \item{Obtain the probability density function (pdf) of the non-blocked time interval by analyzing the alternating renewal process in question;}
  \item{Produce the pdf of the blocked interval by representing it as a busy period in the M/GI/$\infty$ queuing system\footnote{According to the Kendall notation: $M$ is Poisson arrival process, $GI$ is general distribution of service time, and $\infty$ is infinite number of servers.}, where the service time distribution corresponds to the time spent by a blocker in the LoS blockage zone;}
  \item{Calculate all of the metrics of interest, including the moments, the residual time distributions, as well as the conditional non-blocked/blocked state probabilities by applying conventional techniques.}
\end{itemize}

\subsection{LoS Blockage Zone Geometry}

\begin{figure}[!b]
\centering
\includegraphics[width=0.8\columnwidth]{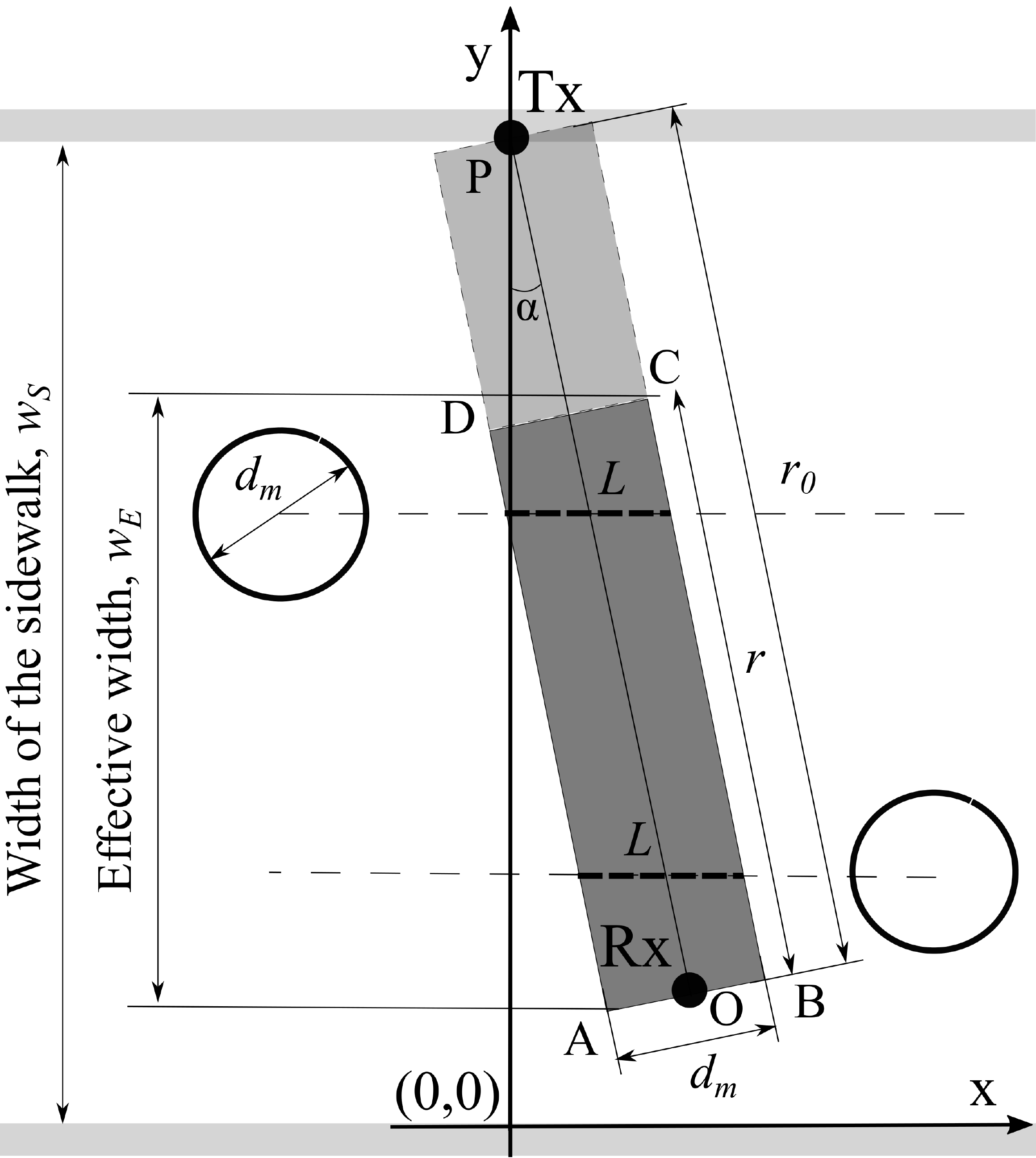}
\caption{Geometry of the LoS blockage zone.}
\label{fig:geometrical_scenario}
\end{figure}

Consider the geometrical scenario represented in Fig.~\ref{fig:geometrical_scenario}. It should be noted that rectangle ABCD, named the LoS blockage zone, is the only area where the presence of a blocker causes blockage of the LoS link. Any blocker which appears outside of this zone (closer to Tx, outside of ABCD) will not affect the LoS link.

Since any blocker entering the LoS blockage zone in question (particularly, the center of a cylinder) occludes the LoS, the width of the zone equals the base diameter of the blocker, $d_{m}$. The length of this zone reflects the maximum possible distance, where the height of the blocker still affects the LoS. As illustrated in Fig.~\ref{fig:geometrical_scenario}, from the geometrical considerations the latter follows as
\begin{align}
r=\frac{r_0(h_B-h_R)}{h_T-h_R} + d_m/2,
\end{align}
where $r_{0}$ is two-dimensional distance between the Tx and the Rx, while $h_{B}$, $h_{R}$, and $h_{T}$ are the heights of the blocker, Rx, and Tx, respectively.

The coordinates of Tx and Rx located at the points P and O (see Fig.~\ref{fig:geometrical_scenario}), respectively, are then given by
\begin{align}\label{eqn:coord1}
&x_P = 0,\qquad{}\qquad{}y_P = w_S, \nonumber \\
&x_O = r_0 \sin(\alpha),\,\,\, y_O = w_S - r_0 \cos(\alpha).
\end{align}

The coordinates of the blockage zone vertices are thus
\begin{align}\label{eqn:coord2}
&x_A = x_O - \frac{d_m}{2} \cos(\alpha),\quad{}y_A = y_O - \frac{d_m}{2} \sin(\alpha), \nonumber \\
&x_B = x_O + \frac{d_m}{2} \cos(\alpha),\quad{}y_B = y_O + \frac{d_m}{2} \sin(\alpha), \nonumber \\
&x_C = x_B - r \cos(\frac{\pi}{2} - \alpha), y_C = y_B + r \sin(\frac{\pi}{2} - \alpha), \nonumber \\
&x_D = x_A - r \cos(\frac{\pi}{2} - \alpha), y_D = y_A + r \sin(\frac{\pi}{2} - \alpha),
\end{align}
where $\alpha$ is the angle characterizing the position of the Rx in relation to the Tx location, as shown in Fig.~\ref{fig:geometrical_scenario}.

\begin{figure*}[!!t]
\begin{align}\label{eqn:busyPeriod}
F_{\eta}(x)=1-\Bigg([1-F_T(x)]\left[1-\int_{0}^{x}(1-F_{\eta}(x-z))\exp(-\lambda F_T(z))\lambda{}dz\right]
+ \int_{0}^{x}(1-F_{\eta}(x-z))|de^{-\lambda F_T(z)}|\Bigg).
\end{align}
\hrulefill
\normalsize
\end{figure*}


\subsection{Renewal process analysis}\label{sec:method}

\begin{figure}[b!]
\centering
\includegraphics[width=0.9\columnwidth]{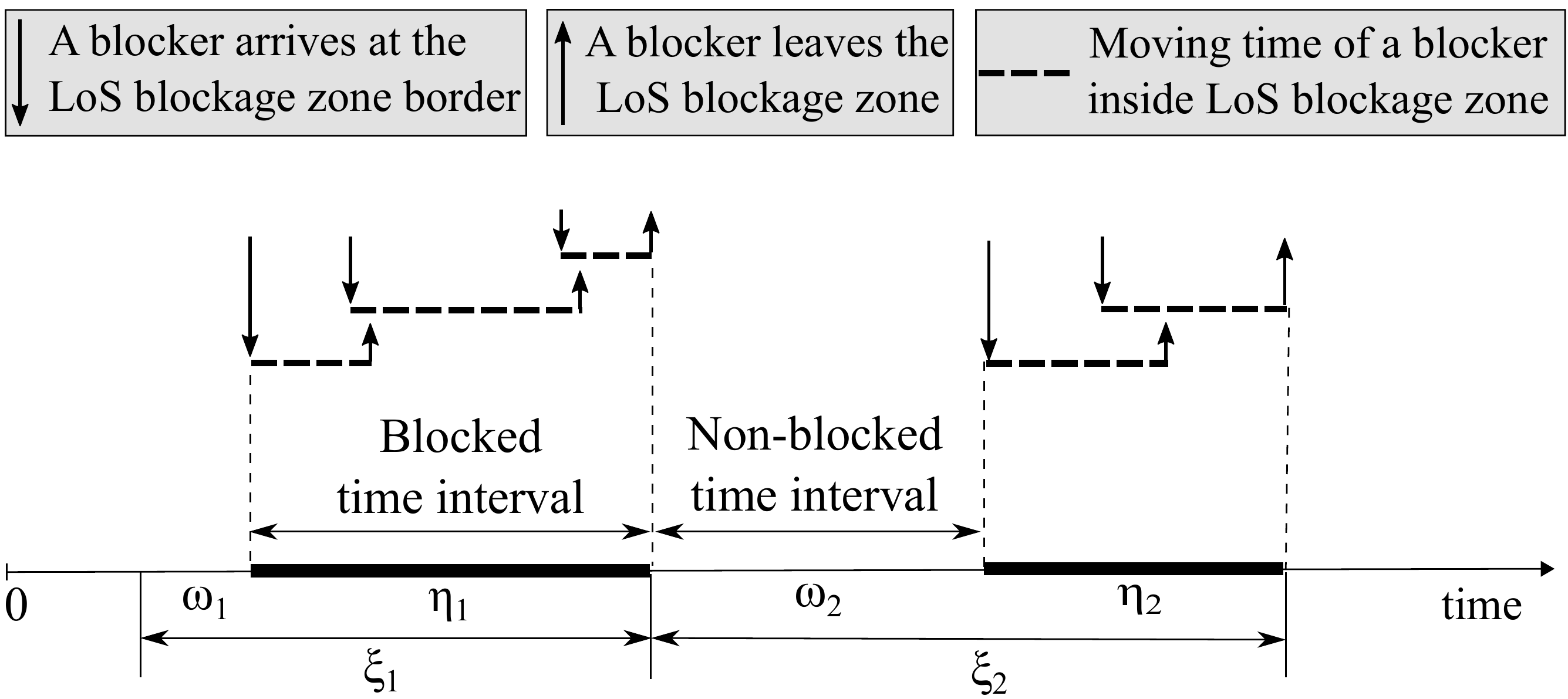}
\caption{Renewal process associated with the LoS blockage, where every blocker might spend different time when occluding the LoS.}
\label{fig:renewal_process}
\end{figure}


Let $\omega_{j}$ and $\eta_{j}$, $j=1, 2, \ldots$, denote the time spent in the non-blocked and blocked intervals, respectively, as shown in Fig.~\ref{fig:renewal_process}. As one may observe, these intervals alternate, that is, non-blocked period always precedes the blocked one and vice versa. Since the entry of blockers into the LoS blockage zone is modeled as a Poisson process, durations of non-blocked and blocked intervals are mutually independent. Hence, the process of the LoS blockage can be modeled as an alternating renewal process, as displayed in Fig.~\ref{fig:renewal_process}. The proposed methodology is valid for all three scenarios of interest.

Define $\xi_{j}=\omega_{j}+\eta_{j}$. The points $0$, $\xi_{1}$, $\xi_{1} + \xi_{2}$, and $\xi_{1} + \xi_{2} + \xi_{3}$ are the renewal moments that form the process at hand. The density of this process follows from \cite{cox_Miller} as
\begin{align}\label{eqn:51}
f(x)=\lambda F_{T}(x)\exp\left(-\lambda \int_{0}^{x}[1-F_{T}(y)]dy\right),
\end{align}
where $\lambda$ is the intensity of the blocker arrivals into the zone, $F_{T}(y)$ is the cumulative distribution function (CDF) of the zone residence time $T=L/V$ for a single blocker, where $L$ is the distance over which a blocker travels inside the blockage zone.


The time spent in the non-blocked part, $\omega_{j}$, follows an exponential distribution with the parameter $\lambda$, $F_{\omega}(x)=1-e^{-\lambda x}$, with the mean $\mathbb{E}[\omega]=1/\lambda$ \cite{cox}. This result follows directly from the fact that the left-hand sides of the time intervals spent in the LoS blockage zone by a single blocker follow a Poisson process in time with the arrival intensity of $\lambda$ per time unit. Therefore, the time period between the end of an interval $\eta_j$ (see Fig.~\ref{fig:renewal_process}), which is considered as an arbitrary point, and the starting point of the next interval $\eta_{j+1}$ is distributed exponentially.

Consider now the blocked interval.
\begin{mydef}
\label{prop2}
Let $F_{\eta}(x)$ be the CDFs of the time in the blocked intervals, $\eta_{j}$, $j=1,2,\dots$, with the mean of $\mathbb{E}[\eta]$. The distribution of the blocked interval, $F_{\eta}(x)$, is the same as the distribution of a busy period in the $M/GI/\infty$ queuing system given by (\ref{eqn:busyPeriod}), see e.g.,~\cite{daley}.
\end{mydef}

\begin{proof}
The proposition is proved by exploiting the analogy with the busy time distribution in $M/GI/\infty$ queuing system. Consider now an evolution of a busy period in $M/GI/\infty$ system. It starts at some $t=t_1$ with a customer arriving into the system. Each arrival during the service time of this customer prolongs the busy period if and only if its service time is greater than the service time of the customers that are currently in service. The busy period ends when a customer upon its departure leaves an empty system. Analyzing the illustration of the renewal process associated with the LoS blockage interval, the analogy with the busy period in $M/GI/\infty$ system is established. Indeed, each blocker extends the LoS blockage period if and only if its blockage time is greater than the blockage time of those blockers currently occluding the LoS. The CDF of the busy period in $M/GI/\infty$ system has been obtained in \cite{daley} and is provided in (\ref{eqn:busyPeriod}).\qedhere
\end{proof}

Note that (\ref{eqn:busyPeriod}) can be evaluated numerically for any $F_{T}(x)$.

\subsection{Residence Time in the LoS Blockage Zone}

To proceed further with deriving the metrics of interest, the CDF of the residence time $T=L/V$ in the LoS blockage zone for a single user is required. Recalling the principles of linear transformation of random variables~\cite{ross}, the pdf of the time $T=L/V$ (for all the scenarios of interest) reads as
\begin{align}
f_T(x)&=Vf_L(xV).
\end{align}
Hence, it is sufficient to find the pdf of distance $L$ that one blocker travels inside the LoS blockage zone, $f_L$, in order to derive $f_T$.
The notation employed in what follows is clarified in Fig.~\ref{fig:geometrical_scenario}. Note that the arrival intensity of the blockers $\lambda$ that enter the LoS blockage zone is different for all the considered scenarios and is derived in what follows by using $\lambda_I$. The latter is the initial arrival intensity of blockers that cross the width of the sidewalk for the first and second scenarios, S1 and S2 (see Fig.~\ref{fig:analytic}). For the sake of the analysis, the park/square scenario, S3, has the arrival intensity of $\lambda_I=\lambda$.

Note that the derivation of distance $L$ is a scenario-specific part of the analysis as it requires a certain distribution of the entry points of blockers to the LoS blockage zone.

\textit{First scenario, S1.} Let $F_{Y}(x), 0 \leq x \leq w_S$, be the CDF of the $y$-coordinate of the entry point for a blocker. Since only the blockers crossing the blockage area are of interest, this distribution is truncated. The resulting distribution $F_{\widetilde{Y}}(x)$ is defined on $y_A \leq x \leq y_C$.

The CDF of the distance $L$ traversed by a blocker in the LoS blockage zone is therefore
\begin{align}\label{eqn:coord3}
\hspace{-0.9em}F_{L}(x) =
\begin{cases}
0, x < 0, \\
F_{\widetilde{Y}}(y_C) - F_{\widetilde{Y}}\left(y_C - x \cos(\alpha) \sin (\alpha)\right)  \\
-\left( F_{\widetilde{Y}}(y_A) - F_{\widetilde{Y}}\left(y_A + x \cos(\alpha) \sin (\alpha)\right) \right),  \\
0 \leq x < x_{min}, \\
1, x \geq x_{min},
\end{cases}
\end{align}
where $x_{min}=\min(d_m/\cos(\alpha),r/\sin(\alpha))$.

For the sidewalk 1 scenario (S1) (see Fig.~\ref{fig:analytic}(a)), (\ref{eqn:coord3}) takes the form of
\begin{align}\label{eqn:coord4}
F_{L}^{1}(x) =
\begin{cases}
0, \qquad{}\qquad{}\qquad{}\qquad{}\,\,x \leq 0, \\
\displaystyle{\frac{x \sin(2\alpha)}{y_C - y_A}}, \,\,\,\,\,\,\,\,0 < x \leq x_{min}, \\
1, \qquad{}\qquad{}\qquad{}\,\,\,\,\,x > x_{min}.
\end{cases}
\end{align}

The arrival intensity of blockers entering the zone that affects the LoS for the first scenario, S1, is delivered as
\begin{align}\label{eqn:544}
\lambda = \lambda_I \frac{w_E}{w_S},
\end{align}
where $w_S$ is the width of the sidewalk, $\lambda_I$ is the arrival intensity of blockers on the width $w_S$, and $w_E = \max (y_A,y_B,y_C,y_D) - \min (y_A,y_B,y_C,y_D)$ is the projection of rectangle $ABCD$ on Y-axis, named the effective width, as shown in Fig.~\ref{fig:geometrical_scenario}.

The residence time in the LoS blockage zone for the second and the third scenario is derived in Appendix~\ref{appendix:appendix_A}.

\vspace{-0.2cm}
\subsection{Metrics of Interest}

\subsubsection{Mean and Fraction of Time in Non-Blocked/Blocked State} The fraction of time in the non-blocked/blocked state can be produced by utilizing the mean time spent in each state, i.e., \cite{cox}
\begin{align}\label{eqn:56_}
&\mathbb{E}[T_{l}]=\frac{\mathbb{E}[\omega]}{\mathbb{E}[\omega]+\mathbb{E}[\eta]},\,\mathbb{E}[T_{n}]=\frac{\mathbb{E}[\eta]}{\mathbb{E}[\omega]+\mathbb{E}[\eta]},
\end{align}
where $\mathbb{E}[\omega]$ and $\mathbb{E}[\eta]$ are the means of the non-blocked/blocked intervals.

Recall that due to the exponential nature of $\omega$, $\mathbb{E}[\omega]=1/\lambda$. The mean $\mathbb{E}[\eta]$ can be obtained numerically by using (\ref{eqn:busyPeriod}). However, there is a simpler approach that is outlined below. Observe that the renewal density $f(x)$ is $f(x)=1/\mathbb{E}[\xi]$, when $t\rightarrow \infty$. From (\ref{eqn:51}), after employing the Laplace transform (LT), we establish that it is also equal to $f(x)=\lambda\exp⁡(-\lambda \mathbb{E}[T])$, where $\mathbb{E}[T]$ is the mean zone residence time for a single blocker, see~\cite{gapeyenkoICC} for details. Hence, the following holds
\begin{align}\label{eqn:541}
\mathbb{E}[\xi]=\frac{1}{\lambda}\exp(\lambda \mathbb{E}[T]).
\end{align}

Then, $\mathbb{E}[\eta]$ can be established as
\begin{align}\label{eqn:155}
\mathbb{E}[\eta]&=\int_{0}^{\infty}[1-F_{\eta}(x)]dx\nonumber\\
&=\int_{0}^{\infty}\left(1-F_{\xi}(x)-\frac{f_{\xi}(x)}{\lambda}\right)dx=\mathbb{E}[\xi]-\frac{1}{\lambda}.
\end{align}

Substituting (\ref{eqn:541}) into (\ref{eqn:155}), we arrive at
\begin{align}\label{eqn:55_}
\mathbb{E}[\eta]=\frac{1}{\lambda}[\exp(\lambda \mathbb{E}[T])-1].
\end{align}

\subsubsection{Residual Time in Non-Blocked/Blocked State} Here, the distribution of the residual time spent in the non-blocked/blocked state given that the user is currently in the non-blocked/blocked state is characterized. Recall that the distribution of the non-blocked interval is exponential, while the CDF for the blocked interval is provided in (\ref{eqn:busyPeriod}). Hence, the residual time distribution in the non-blocked state is also exponential with the same parameter. Therefore, the residual blocked time CDF is
\begin{align}\label{eqn:56}
F_{t_{\eta}}(t)=\frac{1}{\mathbb{E}[\eta]}\int_{0}^{t}[1-F_{\eta}(y)]dy,
\end{align}
and the residual non-blocked time CDF is
\begin{align}\label{eqn:100}
F_{t_{\omega}}(t)=1-e^{-\lambda t},\,t\geq{}0.
\end{align}

\subsubsection{Conditional Non-Blocked/Blocked State Probabilities} Consider now two instants of time, $t_{0}=0$ and $t_{1}$, $t_{1}-t_{0}=t>0$. Denoting the non-blocked and the blocked states by $0$ and $1$, respectively, the conditional probabilities, $p_{00}(t)$, $p_{01}(t)$ as well as $p_{10}(t)$, $p_{11}(t)$ that there is non-blocked/blocked state at $t_{1}$ given that there was non-blocked/blocked state at $t_{0}$ are calculated further. The general solution for this problem follows from \cite{cox} and particularly $p_{00}(t)$ can be established as
\begin{align}\label{eqn:03}
p_{00}(t)= \frac{\mathbb{E}[\omega]}{\mathbb{E}[\omega]+\mathbb{E}[\eta]} + \frac{g(t)}{\mathbb{E}[\omega]},
\end{align}
where $g(t)$ has the LT of
\begin{align}
g^{*}(s) = \frac{\mathbb{E}[\omega]\mathbb{E}[\eta]}{(\mathbb{E}[\omega]+\mathbb{E}[\eta])s} - \frac{(1-f_{\omega}^{*}(s))(1-f_{\eta}^{*}(s))}{s^2(1- f_{\omega}^{*}(s)f_{\eta}^{*}(s))},
\end{align}
where $f_{\omega}^{*}(s)$ and $f_{\eta}^{*}(s)$ are the LTs of $f_{\omega}(x)$ and $f_{\eta}{(x)}$, respectively.

In the target case, the density of the blocked period is not available in a closed form, thus preventing from transitioning to the LT domain. For practical calculations, a simpler approach is proposed below based on utilizing the time domain convolutions. Observe that the probabilities $p_{00}(\Delta{t})$ and $p_{01}(\Delta{t})$ can be represented as
\begin{align}\label{eqn:sum}
&p_{00}(\Delta{t}) = \sum_{i=0}^{\infty}\mathbb{P}\{A_{i}(\Delta{t})\},\nonumber\\
&p_{01}(\Delta{t}) = \sum_{i=1}^{\infty}\mathbb{P}\{B_{i}(\Delta{t})\},
\end{align}
where $A_{i}(t)$ are the events corresponding to starting in the non-blocked interval at $t_{0}$ and ending in the non-blocked interval after some $\Delta{t}=t_{1}-t_{0}$, while having exactly $i$, $i=0,1,\dots$, blocked periods during $\Delta{t}$. Similarly, $B_{i}(t)$ are the events corresponding to starting in the non-blocked interval at $t_{0}$ and ending in the blocked interval at $t_{1}$, while having exactly $i$, $i=1,2,\dots$, non-blocked periods during $\Delta{t}$.

The probability of the event $A_{0}$, which is defined as residing in the non-blocked interval $\omega$ at time $t_1 = t_{0} + \Delta t$ given that the system was in the same non-blocked state $\omega$ at time $t_0$, is produced by
\begin{align}\label{eqn:103}
\mathbb{P}\{A_0(\Delta{t})\} = 1-F_{t_{\omega}}(\Delta{t}),
\end{align}
where $F_{t_{\omega}}(\Delta{t})$ is the residual time in the non-blocked period as obtained in (\ref{eqn:100}).

\begin{figure*}[!t]
\centering
    \subfigure[{Fraction of time spent in the blocked interval for one user}]
    {
        \includegraphics[height=0.3\textwidth]{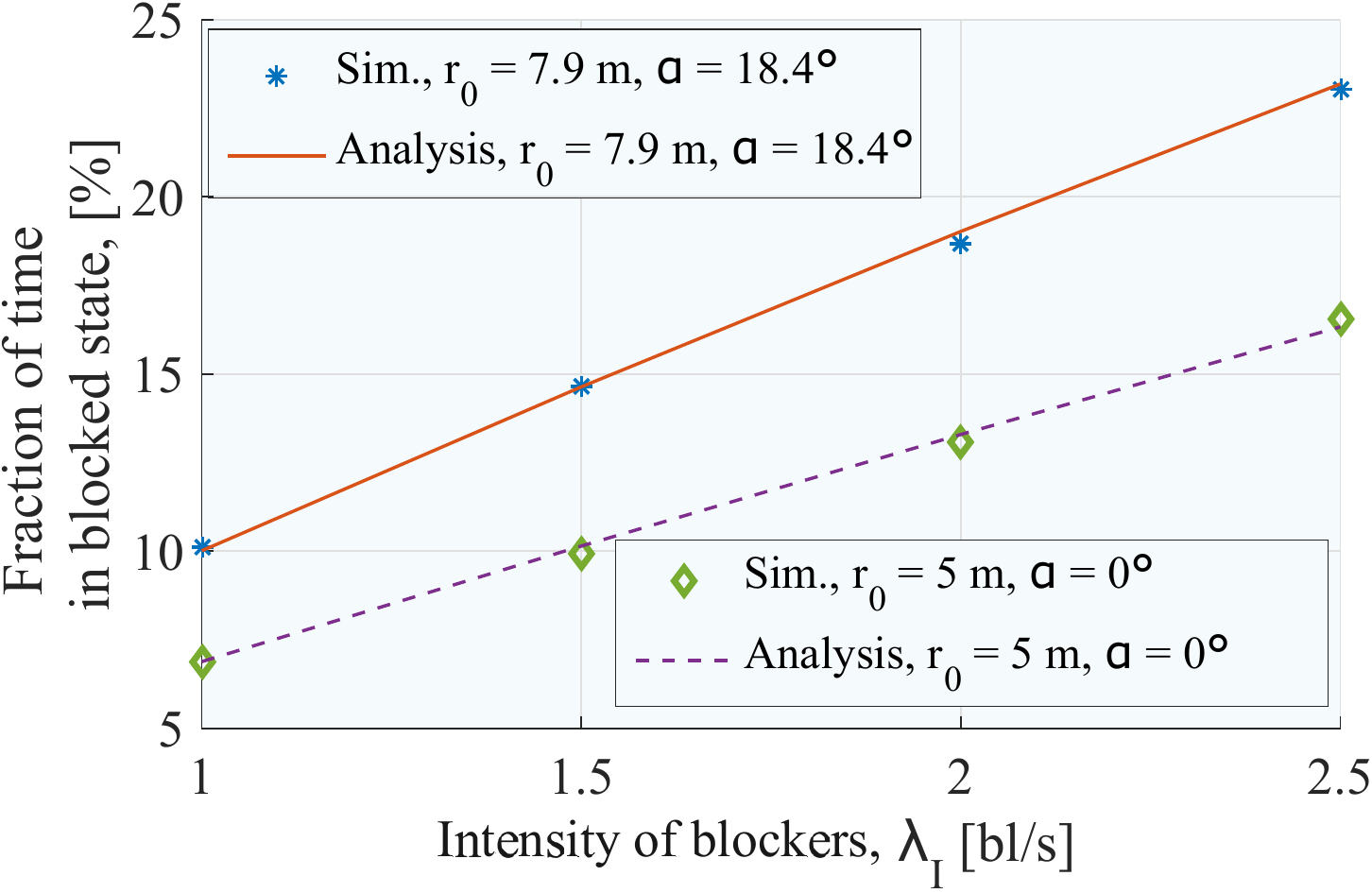}
        \label{fig:av_bl_sls}
    }
    \subfigure[{CDF of time in the blocked interval, $\lambda_{I}$ = 3 bl/s}]
    {
        \includegraphics[height=0.3\textwidth]{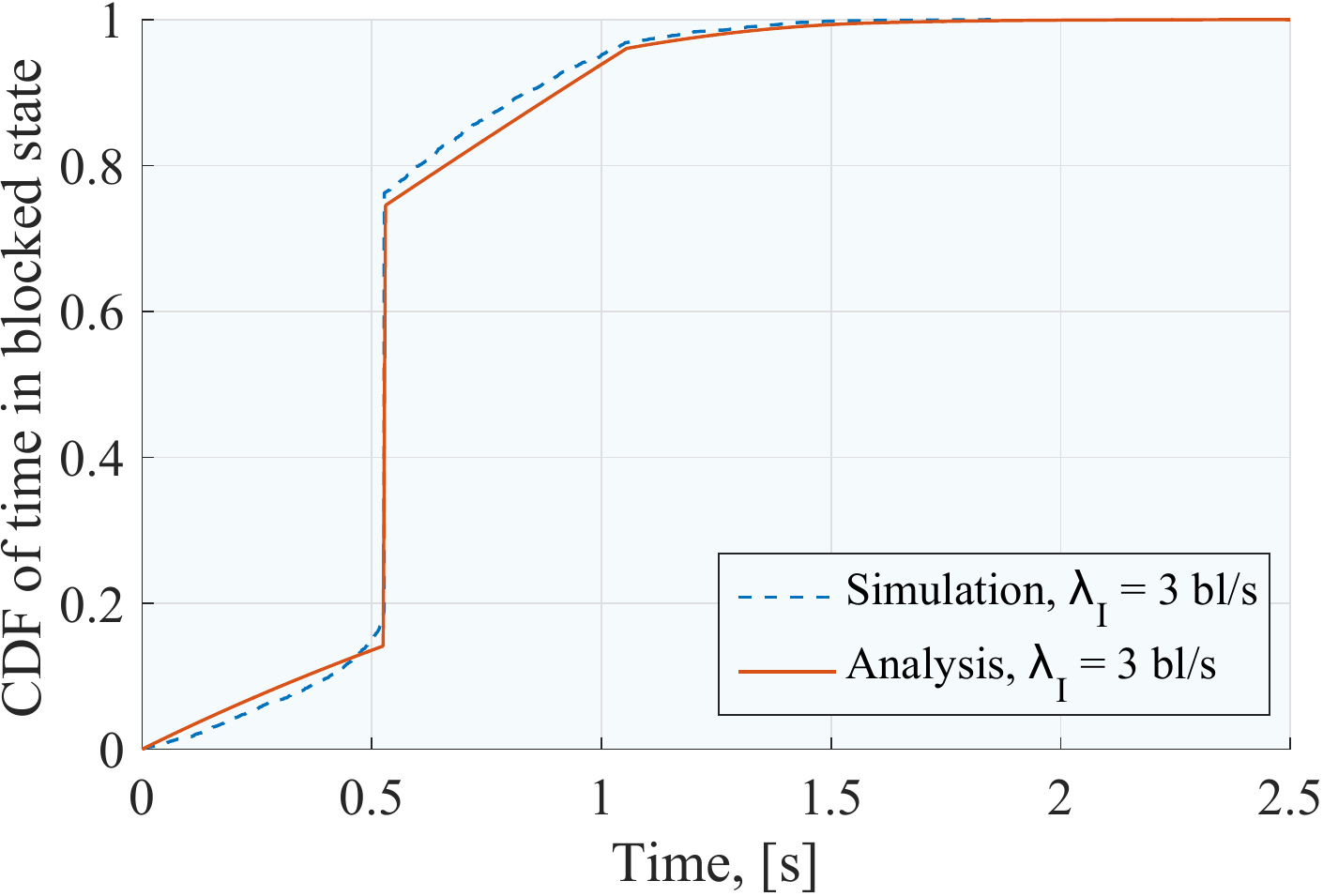}
        \label{fig:cdf_bl_sls}
    }
    \caption{Benchmarking the analytical model against the simulation results for the first usage scenario, S1.}
    \vspace{-0.4cm}
    \label{fig:graphs}
\end{figure*}

The probability of the event $B_{1}$, which is defined as residing in the blocked interval $\eta$ at time $t_{1}$ given that the system was in the preceding non-blocked state $\omega$ at time $t_{0}$, is
\begin{align}\label{eqn:104}
\hspace{-1em}\mathbb{P}\{B_1(\Delta{t})\} &= 1 - F_{\eta + t_{\omega}}(\Delta{t}) - (1-F_{t_{\omega}}(\Delta{t}))\nonumber\\
&= F_{t_{\omega}}(\Delta{t}) - F_{\eta + t_{\omega}}(\Delta{t}),
\end{align}
where $F_{\eta}$ is the CDF of the blocked interval from (\ref{eqn:busyPeriod}) and $F_{\eta+\omega}$ denotes the CDF of the sum of random variables $x$ and $y$. As it was stated previously, note that the random variables $\eta$, $\omega$, and $t_{\omega}$ are independent. Hence, the CDF of the sum $F_{\eta+\omega}$ is obtained by convolving the densities of $x$ and $y$ and then integrating from $0$ to $x$.

Consider the event $A_{1}$ corresponding to when the Rx is in the non-blocked interval at time $t_{1}$ given that it was in the preceding non-blocked interval at time $t_{0}$ (there is a blocked interval embedded in between $t_{0}$ and $t_{1}$). The probability of this event is
\begin{align}
\mathbb{P}\{A_1(\Delta{t})\} = F_{\eta + t_{\omega}}(\Delta{t})
- F_{\omega + \eta + t_{\omega}}(\Delta{t}),
\end{align}
where $F_{\omega}$ is the CDF of the non-blocked interval.

Further, the probability of the event $B_{2}$ that the Rx is in the blocked interval at $t_{1}$ given that it was in the preceding blocked interval at $t_{0}$ (there is an additional non-blocked interval embedded in between $t_{0}$ and $t_{1}$), is established as
\begin{align}
\hspace{-1em}\mathbb{P}\{B_2(\Delta{t})\} = F_{\omega + \eta + t_{\omega}}(\Delta{t}) - F_{\eta + \omega + \eta + t_{\omega}}(\Delta{t}).
\end{align}


Finally, the following is obtained
\begin{align}\label{eqn:104_Ai}
\hspace{-1em}\mathbb{P}\{A_i(\Delta{t})\} &= F_{\sum_{j = 1}^{i - 1}(\eta + \omega) + \eta + t_{\omega}}(\Delta{t}) \nonumber\\
& - F_{\sum_{j = 1}^{i}(\eta + \omega) + t_{\omega}}(\Delta{t}), \qquad{}i \geq 1,\nonumber\\
\hspace{-1em}\mathbb{P}\{B_i(\Delta{t})\} &= F_{\sum_{j = 1}^{i - 1}(\eta + \omega) + t_{\omega}}(\Delta{t}) \nonumber\\
& - F_{\sum_{j = 1}^{i - 1}(\eta + \omega) + \eta +  t_{\omega}}(\Delta{t}),\quad{}i \geq 1.
\end{align}

Note that the sum in (\ref{eqn:sum}) is infinite, and the probabilities $p_{00}$ and $p_{01}$ are numerically approximated by summing the terms up to the next summand that is sufficiently close to zero, until when the desired accuracy is achieved. The probabilities $p_{10}(\Delta{t})$ and $p_{11}(\Delta{t})$ are obtained similarly. 


\section{Accuracy Assessment and Numerical Analysis}\label{sec:results}

In this section, the accuracy of the proposed model is assessed by benchmarking against system-level simulations. Then, the extent of temporal dependence under study is characterized as a function of the input parameters.

\vspace{-0.2cm}
\subsection{Accuracy Assessment}

In Fig.~\ref{fig:graphs}, the benchmarking of the proposed analytical model is conducted by utilizing our in-house simulation framework developed specifically for the purposes of this study. For the sake of exposition, it is assumed that the location of the user device of interest is fixed. The initial number of deployed blockers is calculated based on the arrival intensity of blockers entering the width of the sidewalk, $\lambda_I$. Particularly, considering the first usage scenario, S1, whenever the simulation is started, new blockers appear at the sidewalk edge of length $w_S$ according to a Poisson process with the arrival intensity of $\lambda_I$. Blockers then move around across the deployment with the constant speed up to the edge of the deployment area.

\begin{table}[!b]\footnotesize
\centering
\vspace{-0.3cm}
\caption{Baseline system parameters}
\begin{tabular}{p{5.9cm}p{2.1cm}}
\hline
\textbf{Parameter}&\textbf{Value} \\
\hline
Height of Tx, $h_T$&$3$~m\\
Height of Rx, $h_R$&$1.3$~m\\
Tx-Rx distance, $r_0$&$4.6$~m\\
Height of a blocker, $h_B$ & $1.7$~m\\
Diameter of a blocker, $d_m$ & $0.5$~m\\
Speed of a blocker, $V$ & $1$~m/s\\
Width of the sidewalk, $w_S$ & $5$~m\\
Angle, $\alpha$ & $\pi/6$\\
Frequency & $28$~ GHz\\
Bandwidth, $B$ & $1$~ GHz\\
Noise level & $-84$~ dBm\\
Transmit power & $30$~ dBm\\
\hline
\end{tabular}
\vspace{-0.3cm}
\label{tab:sim_parameters}
\end{table}

Fig.~\ref{fig:av_bl_sls} reports on the average user blockage time for the first scenario (S1) obtained by using simulations as well as produced with the proposed analytical model for the width of the sidewalk, $w_S = 10$ m where the remaining parameters are given in the plot and Table \ref{tab:sim_parameters}. The target accuracy was set to $10^{-4}$ which required from 6 to 9 summands in (\ref{eqn:sum}) to achieve it. As one may observe, the analytical results agree well with the simulation data, while both increase linearly with the growing arrival intensity of blockers. To assess the time correlation in the non-blocked/blocked state, the CDF of blocked duration is displayed in Fig.~\ref{fig:cdf_bl_sls}, where the width of the sidewalk is taken as $w_S = 10$~m, $r_0 = 7.9$~m, and $\alpha = 18.4^{\text{o}}$ with the rest of the parameters given in the Table~\ref{tab:sim_parameters}. 

\begin{figure}[!b]
\vspace{-0.2cm}
\centering
\includegraphics[width=0.9\columnwidth]{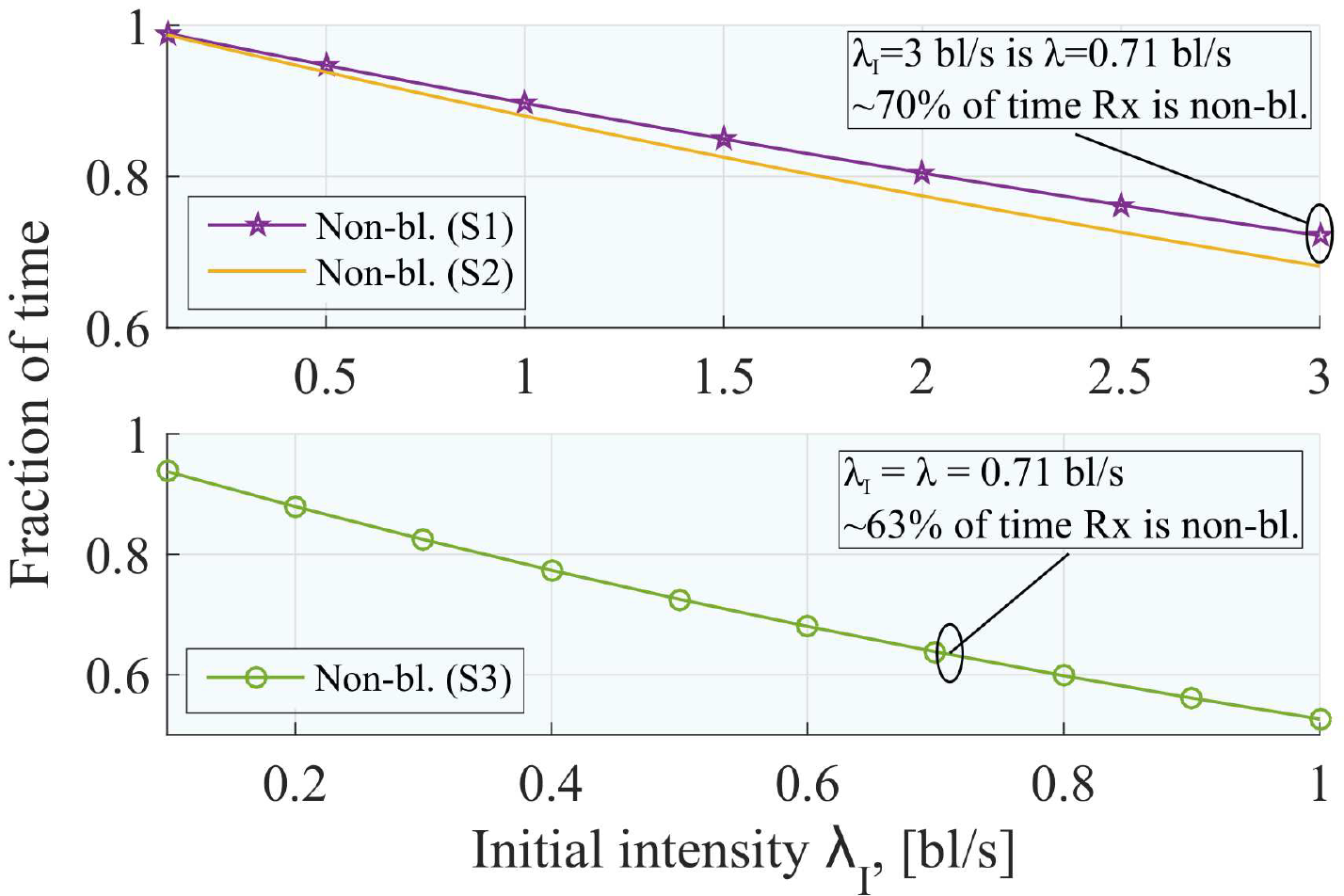}
\vspace{-0.2cm}
\caption{Fraction of time in non-blocked/blocked state as a function of $\lambda_I$.}
\label{fig:fraction}
\end{figure}

\begin{figure*}[t!]
\centering
\subfigure[\label{fig:mean_nLoS}{Mean time in blocked interval}]
{\includegraphics[width=.28\textwidth]{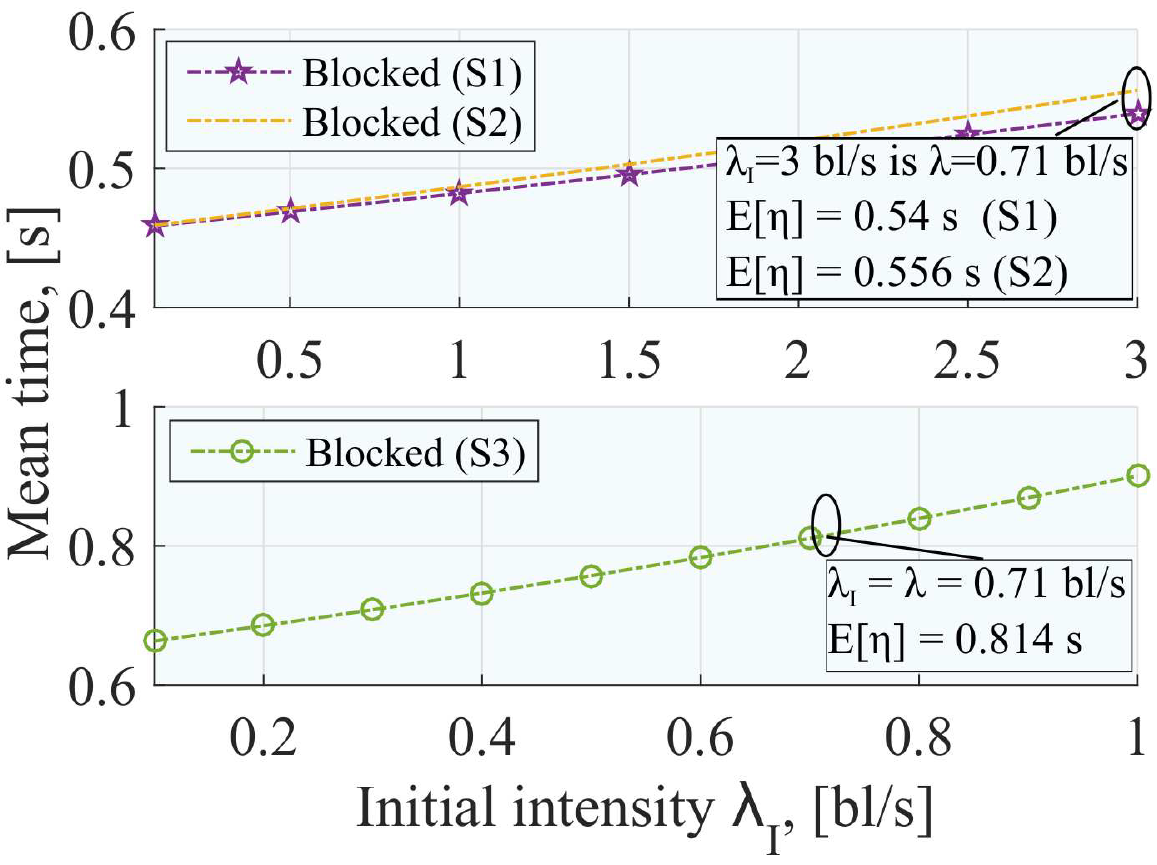}}~~~~~~~
\subfigure[\label{fig:mean_LoS}{Mean time in non-blocked interval}]
{\includegraphics[width=.28\textwidth]{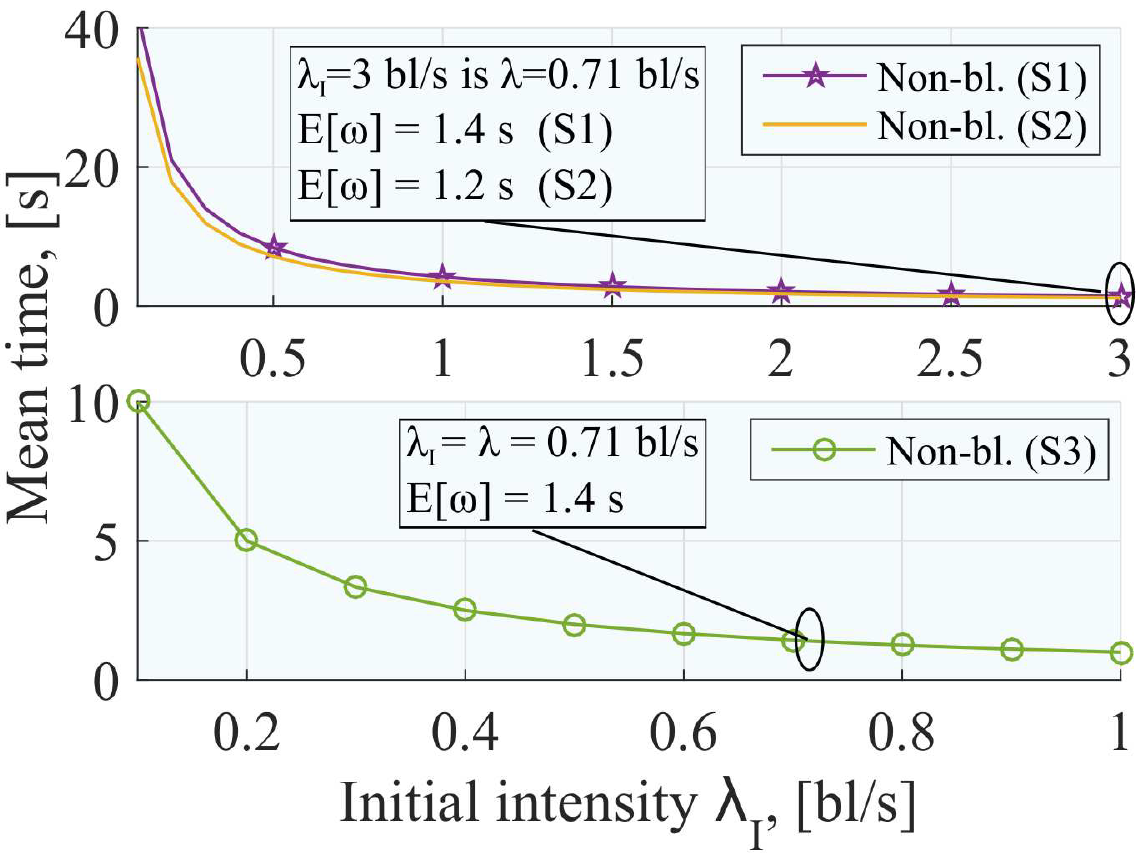}}~~~~~~~
\subfigure[\label{fig:resid_nlos}{CDF of residual blocked time}]
{\includegraphics[width=.28\textwidth]{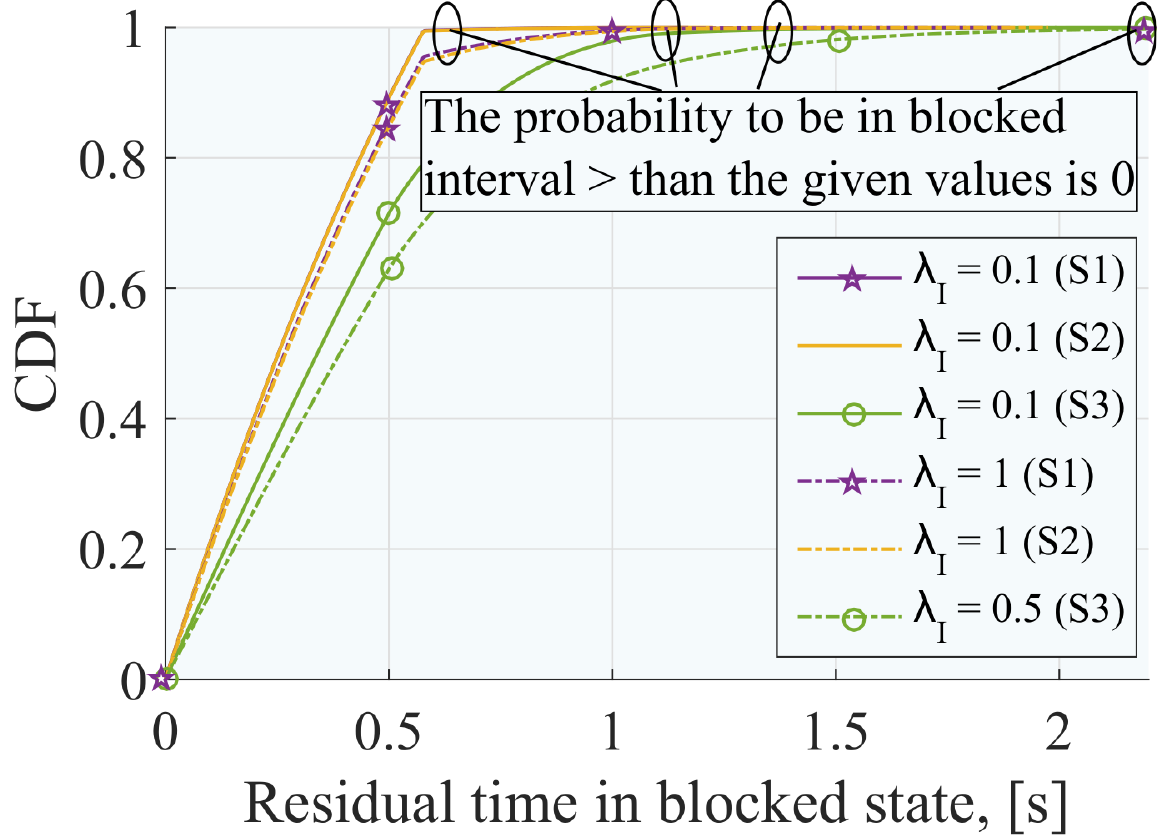}}
\centering
\subfigure[\label{fig:resid_los}{CDF of residual non-blocked time}]
{\includegraphics[width=.28\textwidth]{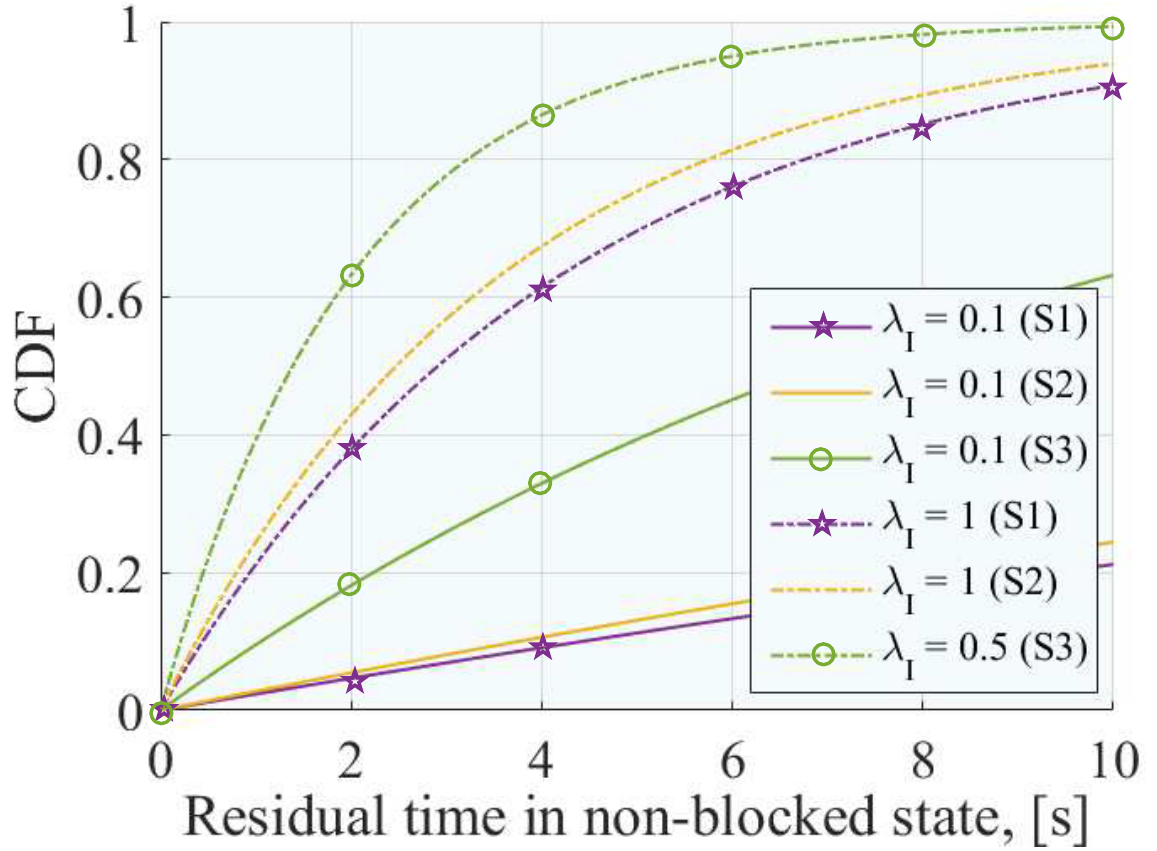}}~~~~~~~
\subfigure[\label{fig:P10_P11}{Non-blocked/blocked state conditioned on blocked state at $t_0$}]
{\includegraphics[width=.28\textwidth]{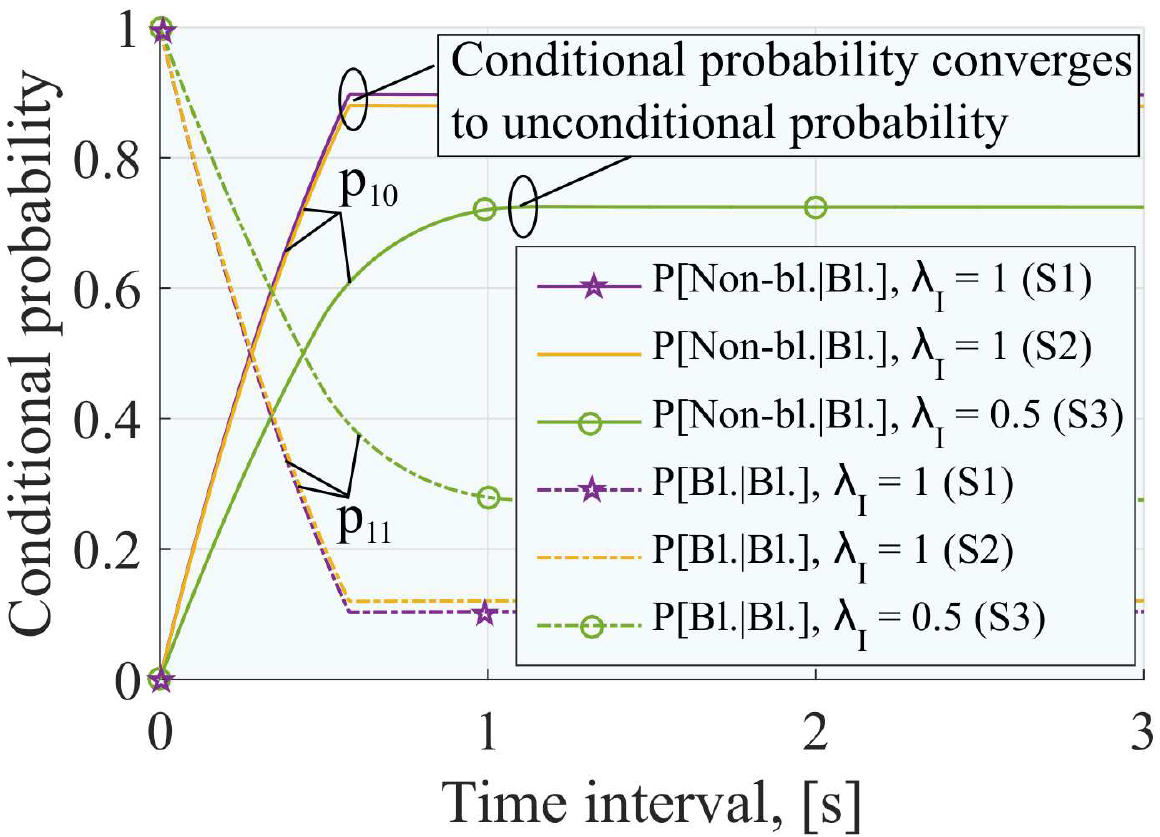}}~~~~~~~
\subfigure[\label{fig:P00_P01}{Non-blocked/blocked state conditioned on non-blocked state at $t_0$}]
{\includegraphics[width=.28\textwidth]{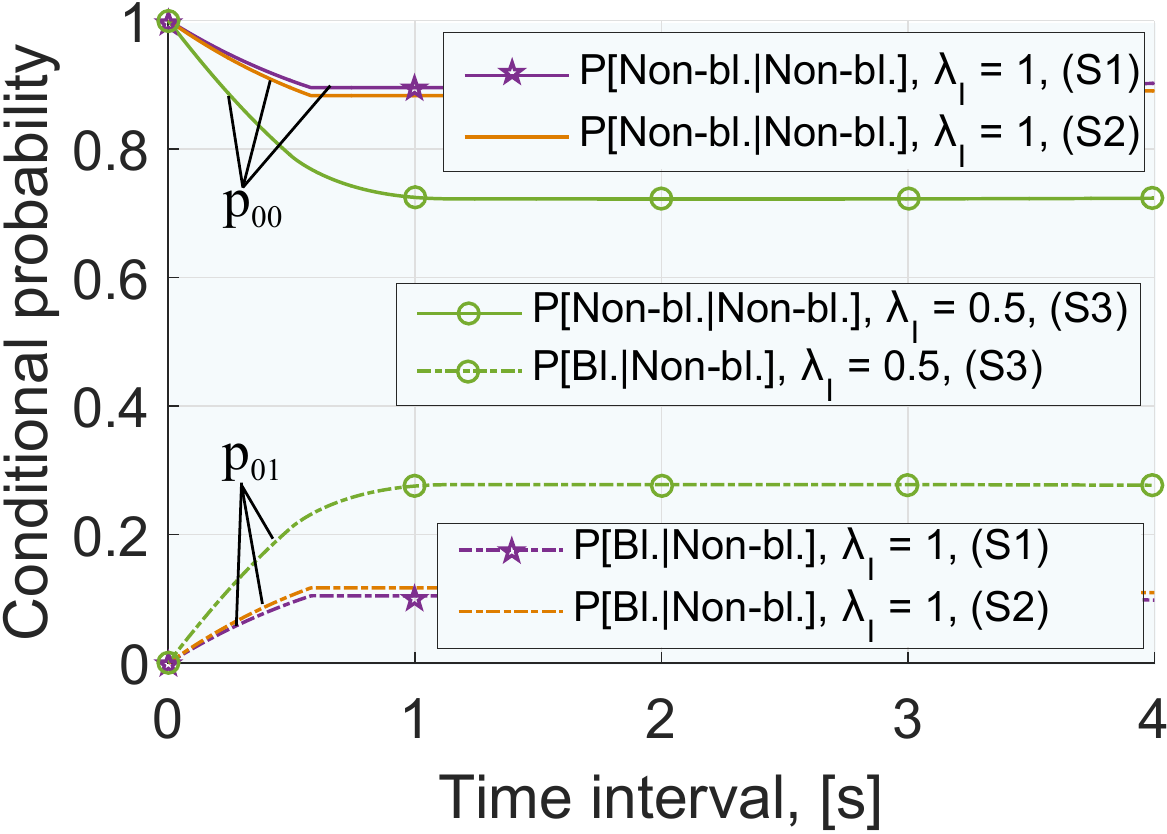}}
\caption{Mean and residual time in non-blocked/blocked state and conditional probability of non-blocked/blocked state given that Rx is non-blocked/blocked.}
\vspace{-0.4cm}
\label{fig:selectedPars}
\end{figure*}

Here, close match between the analytical and the simulation results is also clearly visible. Small discrepancy between simulation and analysis is caused by the specifics of the analytical model. Particularly, in simulations the LoS blockage zone is explicitly modeled by taking into account the circular nature of the blocker. In the developed mathematical model, the LoS blockage zone is assumed to be of rectangular shape thus neglecting the curvature caused by the blocker, see Section~\ref{sec:model} for details. Even though one could extend the model to the case of more complex geometry of the LoS blockage zone thus leading to more complex expressions, the resulting error of approximation by a rectangle is negligible. Also, note that the steep behavior of the CDF around $0.5$~s is explained by the fact that for these particular environmental parameters and intensity of blockers most of the busy periods on the associated M/GI/$\infty$ queue are caused by a single blocker.

\vspace{-0.4cm}
\subsection{Numerical Analysis}


Further, the response of the blockage-related metrics to the selected ranges of input mmWave system parameters is analyzed. It should be mentioned that the choice of values for the parameters, and especially the Tx-Rx distance, is according to the need to compare all three mobility models. However, all parameters are adjustable during the computation if needed. Therefore, the key performance indicator in the deployment of interest is considered, namely, the fraction of time spent in the non-blocked state as illustrated in Fig.~\ref{fig:fraction}. It is a function of the arrival intensity of blockers, $\lambda_{I}$, for all the three scenarios under study. It should be noted that the fraction of time in blocked state is the complement of the fraction of time in the non-blocked state. The parameters for scenarios that are collected in the plot are shown in Table \ref{tab:sim_parameters}.


For the purposes of a numerical comparison, consider the initial intensities for the first (S1) and the second scenario (S2) to be equal to $1$ and $3$ blockers per second, respectively. This corresponds to the following intensities of entering the LoS zone: $0.24$ and $0.71$ blockers per second. The initial arrival intensity is equal to the intensity of entering the LoS zone for the third scenario (S3). Clearly, as the arrival intensity of blockers grows, the fraction of time spent in the non-blocked/blocked state decreases/increases correspondingly. The main observation here is that the resulting trend is close to linear. One may notice further that for the arrival intensity of $0.24$ blockers per second the fraction of time spent in the non-blocked state is almost the same for all the three scenarios. As the arrival intensity increases and approaches the value of $0.71$, the difference between the first two scenarios and the third scenario becomes more considerable.

\begin{figure*}[!t]
\centering
    \subfigure[Varying the cell radius]
    {
	  \includegraphics[width=0.4\textwidth]{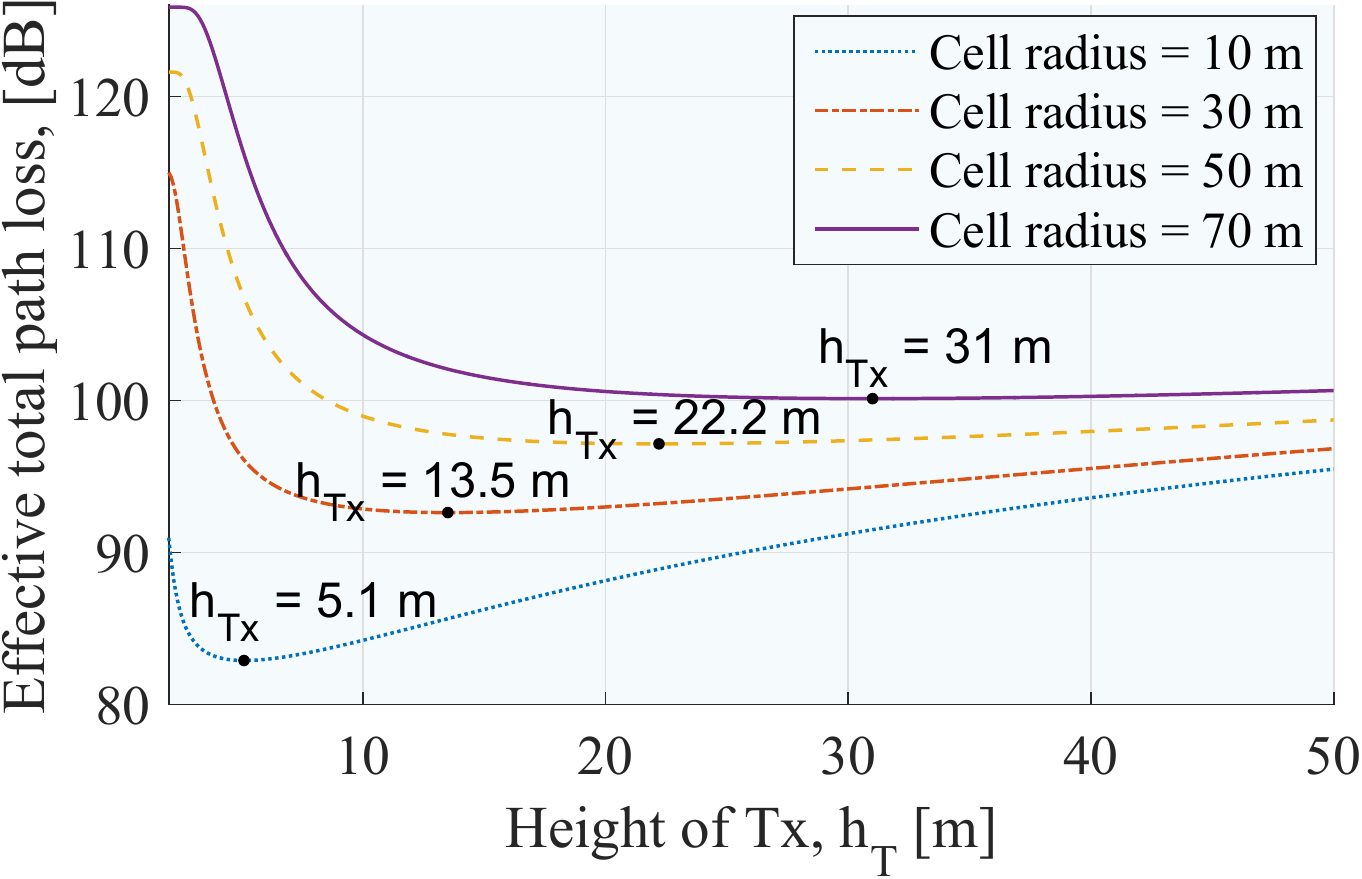}
        \label{fig:height}
    }~~~~~~~~~
   \subfigure[Varying the arrival intensity of blockers]
    {
	  \includegraphics[width=0.4\textwidth]{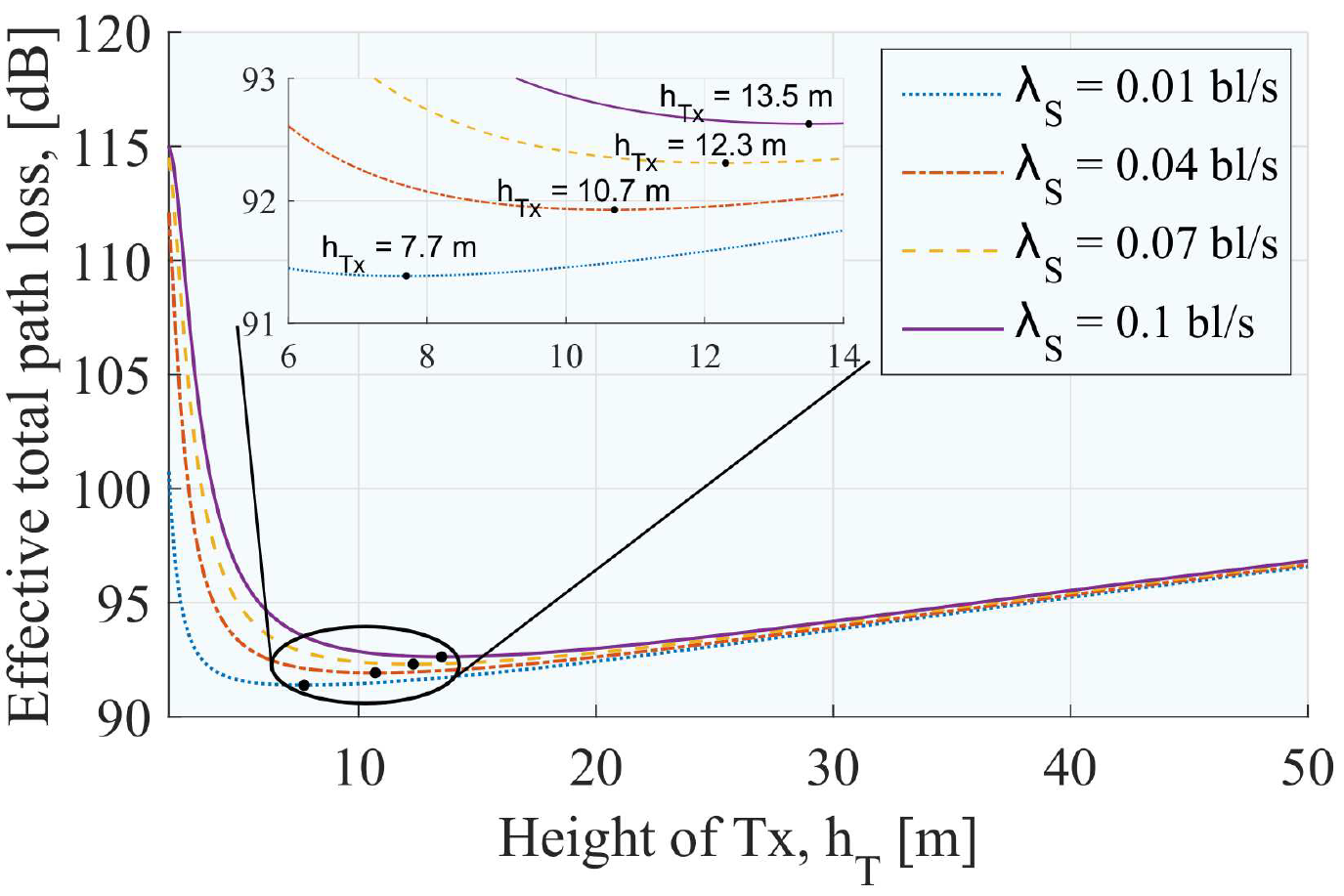}
	  \label{fig:height1}
    }
    \caption{Optimal height of the mmWave AP vs. cell radius and arrival intensity of blockers for the third scenario, S3.}
    \label{fig:height_}
    \vspace{-2mm}
\end{figure*}


Fig.~\ref{fig:mean_nLoS} and~\ref{fig:mean_LoS} report on the absolute values of the mean time spent in the non-blocked/blocked state for the same input parameters. As one may observe, the mean time spent in the non-blocked state decreases significantly as $\lambda_{I}$ increases from $0$ to $3$. However, the difference between the three scenarios for the two selected intensities is not significant. The mean time spent in the blocked state is the longest for the third scenario. It may be explained by the fact that the possible walking distance of a blocker is higher in this third scenario, since the blocker can move closer to the diagonal of the rectangle. Note that the blocked state behavior does not change drastically over a wide range of the considered blocker intensities. This is because even for higher intensities of blockers the blocked intervals are likely to feature only a single blocker occluding the LoS. 


The mean time spent in the non-blocked/blocked state together with the associated fractions produce a direct implication on the dimensioning of mmWave systems. More specifically, using an appropriate propagation model, such as the one presented in~\cite{rappaport}, as well as accounting for the set of the modulation and coding schemes, one can evaluate the average throughput of a user located at a certain distance from the mmWave AP over a particular time slot. Given a certain value of the target mean data rate at the input, this information can be used further for determining the optimal coverage of a single mmWave AP. A close match with the result in \cite{Jacob_10} in terms of the mean time of LoS link blockage under the corresponding values of parameters is noted.


In Fig.~\ref{fig:resid_nlos} and~\ref{fig:resid_los}, the CDFs of the residual time in the blocked/non-blocked state are shown. As one may observe, the probability for the time spent in the blocked interval to exceed the blocker's mean residence time is rather small. For example, the mean time in the blocked interval for the first scenario with $\lambda_I=1$ bl/s is about $0.5$ s and the CDF $F_{t_{\eta}}(t<0.5)=0.9$ approximately. This fact implies that for a wide range of the considered intensities, in most cases, the blocked interval coincides with the residence time of a single blocker. Therefore, a user enters the non-blocked state after a certain time interval, which mainly depends on the size and the speed of the blocker, and less so on the arrival intensity of the blockers (note that the mean time in the blocked interval, Fig.~\ref{fig:mean_nLoS}, and the CDF of the residual time in the blocked state, Fig.~\ref{fig:resid_nlos}, do not change drastically with increasing arrival intensity of blockers). Generally, knowing that the Rx is in the blocked interval, one can estimate the remaining time in this period. This may reduce the amounts of signaling information required for tracking the state of mmWave receivers. Also, the shape of the CDF curves for the residual time in the blocked interval is explained by the particular behavior of the CDF of time for a single blocker movement inside the LoS blockage zone, which has a distinct plateau.


Fig.~\ref{fig:P10_P11} and~\ref{fig:P00_P01} illustrate the behavior of the conditional probability to be in the non-blocked/blocked state at time $t_{1}$ given that the Rx was in the non-blocked/blocked state at time $t_0=0$, $t_{1}>t_{0}$. Due to the long average time in the non-blocked state as compared to the average time in the blocked state, the probability to change the state from non-blocked to blocked is rather small for the considered values of $t_{1}$. In contrast, the probability to become non-blocked given that the Rx was blocked at time $t_0$ increases significantly. After that, the conditional probability converges to the unconditional one and the process in question ``loses'' its memory.

\section{Example Applications of The Methodology}\label{sect:apps}

This section first summarizes two important analytical results stemming from the direct application of the proposed methodology. Then, the achievable performance gains in terms of the computation complexity are demonstrated after applying the model for system-level evaluation of mmWave systems.

\subsection{Optimal Height of the mmWave AP}

Let us first determine the height of the mmWave AP, such that the average path loss to the user is minimal. To this end, the blocker mobility model to estimate the fraction of time in the non-blocked state as a function of $h_{T}$ is utilized, and then the mmWave propagation model from~\cite{rappaport} to characterize the mean path loss as a function of $h_{T}$ is applied.

The average path loss can thus be established as in~\cite{gapeyenkoICC}
\begin{align}\label{eqn:1}
L_{e}=\mathbb{E}[T_l]L_{LoS} + (1-\mathbb{E}[T_l])L_{nLoS},
\end{align}
where $\mathbb{E}[T_l]$ is the fraction of time that the Rx spends in the non-blocked state, which has been derived in (\ref{eqn:56}), while $L_{LoS} = 61.4 + 20\log_{10}(d)$ and $L_{nLoS}=72 + 29.2\log_{10}(d)$ are the path loss values for the LoS and the nLoS components for 28~GHz as obtained in~\cite{rappaport} and $d=\sqrt{(h_T -h_R)^{2} + r_{0}^{2}}$ is the three-dimensional  distance between Tx and Rx.

\begin{figure*}[!t]
\centering
    \subfigure[Varying the cell radius and blocker arrival intensity]
    {
	  \includegraphics[width=0.41\textwidth]{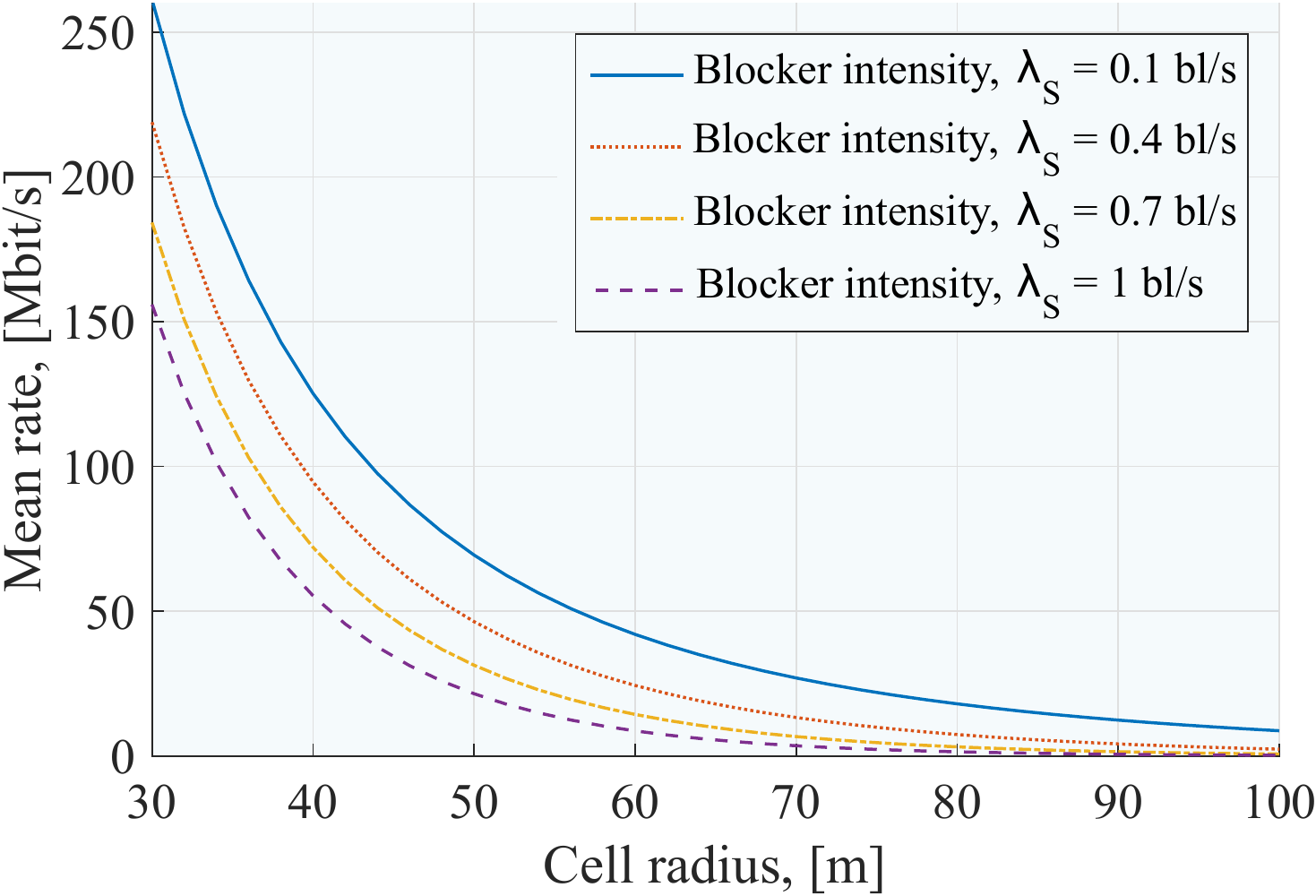}
	  \label{fig:rate}
    }~~~~~~~~~
   \subfigure[Varying the cell radius and user density]
    {
	  \includegraphics[width=0.41\textwidth]{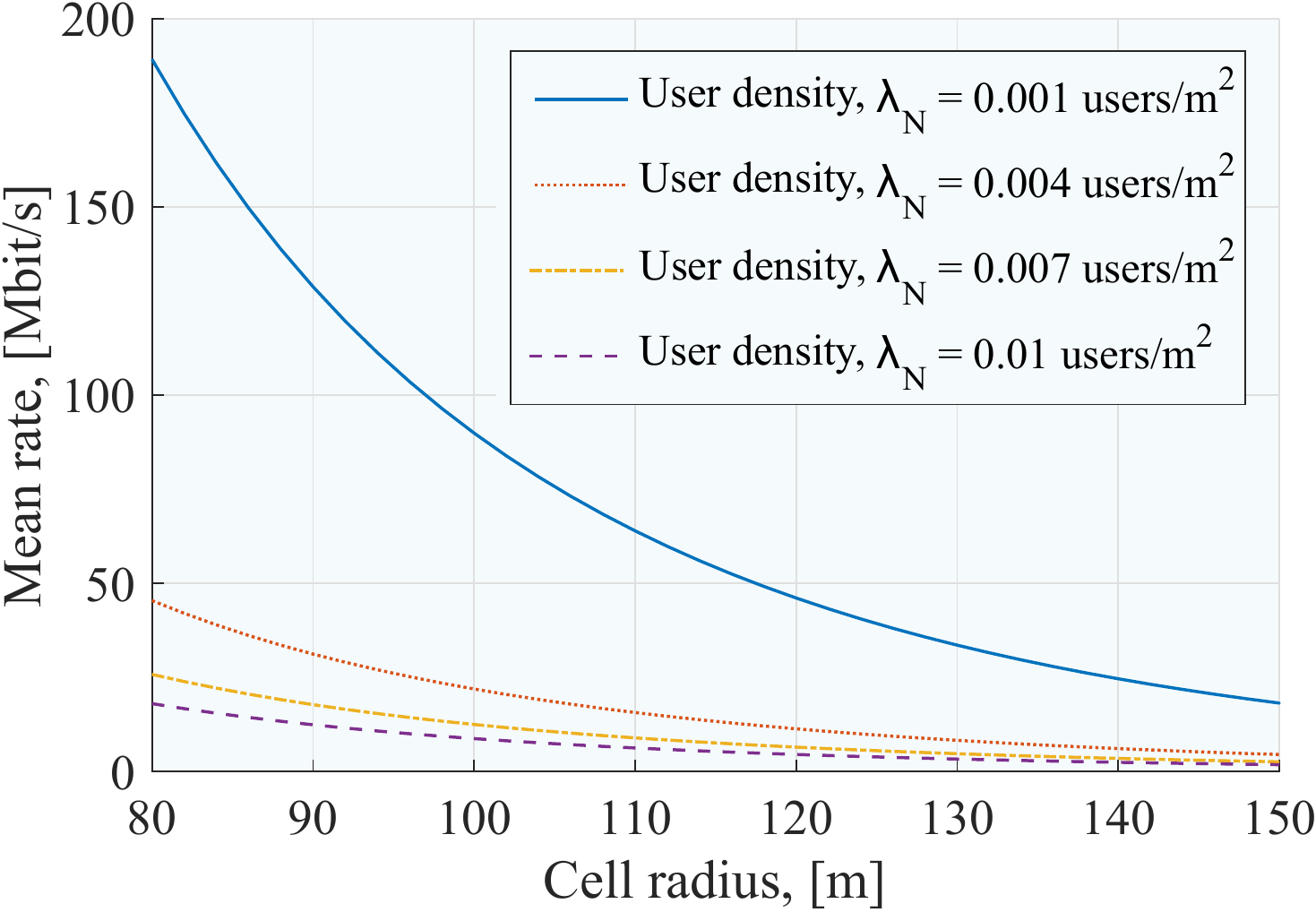}
	  \label{fig:rate1}
    }
    \caption{Mean data rate of the user at the cell edge vs. cell radius, user densities, and blocker arrival intensities for the third scenario, S3.}
    \vspace{-0.3cm}
    \label{fig:rate_}
\end{figure*}

For any value of the arrival intensity of blockers, the optimal height of the mmWave AP within the range of reasonable values can now be established by utilizing the graphical approach and plotting (\ref{eqn:1}) to identify the value minimizing the average path loss at the cell edge (two-dimensional distance between Tx and Rx, $r_0$, is equal to the cell radius in that particular problem). The same could be derived numerically by taking a derivative of the average path loss from (\ref{eqn:1}). To vary the arrival intensity of blockers that enter the LoS blockage zone proportionally to its dimensions, it is assumed that $\lambda_S=0.1$~bl/s is the arrival intensity of blockers crossing the unit square area. The intensity of blockers entering the LoS blockage zone can then be written as $\lambda = \lambda_S r d_m$.

Fig.~\ref{fig:height} demonstrates the optimal Tx height for different cell radius values in the third scenario. Here, the constant arrival intensity of blockers is set to $\lambda_S=0.1$~bl/s, while the remaining parameters are given in Table~\ref{tab:sim_parameters}. As one may learn, an increase in the cell radius requires the mmWave Tx to be deployed higher in order to achieve the optimized propagation conditions at the cell edge. Further, Fig.~\ref{fig:height1} shows the optimal height of the Tx for a fixed cell radius of $30$~m and different intensities of blockers that enter the unit square area of the LoS blockage zone, $\lambda_S$. It could be noticed that with the growing blocker arrival intensity the optimal height of the mmWave Tx increases as well. This effect is explained by the fact that the probability of residing in the non-blocked state decreases; hence, one needs to increase the height of the Tx to maximize the fraction of time spent in the non-blocked state.

The impact of the cell radius and intensity of blockers on the optimal height of the AP is summarized as follows:
\begin{itemize}
	\item{The optimal height of the AP from the range of realistic values highly depends on the cell radius, e.g., after increasing the cell radius by 7 times the optimal height grows by approximately 6 times.}
  \item{The impact of the intensity of blockers on the optimal height is less pronounced. For example, after increasing the intensity of blockers by 10 times the optimal height grows by only 1.7 times.}
\end{itemize}

\subsection{Cell Range Analysis}

Another direct application of the proposed model is to determine the maximum coverage range of the mmWave AP, such that a certain average data rate is delivered to all of the users. The latter can be achieved by ensuring that the user data rate at the cell edge is higher than the required minimum.

Assume a Poisson field of users in $\Re^{2}$ with the density of $\lambda_N$ users per square unit. Let $x$ be the intended radius of the mmWave coverage zone. The number of users in this covered area follows a Poisson distribution with the parameter $\lambda_N\pi{}x^{2}$. The traffic model is considered here to be ``full buffer'', that is, the mmWave system is observed in the highly-loaded conditions. Further, the maximum radius $x$ is determined, such that the capacity of at least $k$ Mbps is provided to each user. The bandwidth of the mmWave system, $B$, is allowed to be infinitesimally divisible. For simplicity, an equal division of bandwidth between all of the users is considered, even though any reasonable resource allocation mechanism can in principle be assumed, e.g., max-min or proportional fair~\cite{walrand2013congestion}.

The capacity delivered to the mmWave Rx located at $x$ can be derived as
\begin{align}
R(x)=cB_{i}\log[1+S(x)],
\end{align}
where $B_{i}$ is the bandwidth made available to the user of interest, $S(x)$ is the average signal-to-noise ratio (SNR) at this user, and $c$ is a constant accounting for imperfections of the modulation and coding schemes. In what follows, $c=1$ is taken for simplicity.

Since the radio resource in the system is assumed to be distributed equally between all of the users, the bandwidth share actually available\footnote{Note that instead of equal division of the bandwidth, more sophisticated resource allocation strategies can be enforced providing a certain degree of trade-off between fairness of per-user rates and aggregate system capacity, e.g., max-min, proportional fairness, weighted $\alpha$-fairness, see \cite{yaacoub2012survey}.} to the Rx located at $x$ is $B_{i}=B/N$, where $N$ is a discrete random variable (RV) having a Poisson distribution with the density of $\lambda_N\pi{}x^{2}$ per considered area of interest. To obtain the SNR $S(x)$, the mmWave propagation model in~\cite{rappaport} is employed by defining as $s_{0}(x)$ the SNR associated with the LoS state and as $s_{1}(x)$ the SNR associated with the nLoS state. The aggregate SNR is a two-valued discrete RV taking the values of $s_{i}(x)$, $i=0,1$, with the probabilities corresponding to the fraction of time spent in the non-blocked ($\mathbb{E}[T_l]$) and blocked ($1-\mathbb{E}[T_l]$) state, respectively. The RVs $B_{i}$ and $S(x)$ are independent and their joint probability mass function (pmf) is derived as the product of the individual pmfs. 

Once this joint pmf is obtained, one may proceed with determining the mean capacity $R(x_c)$ that is provided to a user located at the cell edge $x_c$ as
\begin{align}
\mathbb{E}[R]&=\sum_{N=1}^{\infty}\frac{(\lambda_N\pi{}x_{c}^{2})^N(e^{-\lambda_N\pi{}x_{c}^{2}})}{N!}\nonumber\\
&\times\Big(\mathbb{E}[T_l]c\frac{B}{N}\log[1+s_{0}(x_{c})] \nonumber\\
&+(1-\mathbb{E}[T_l])c\frac{B}{N}\log[1+s_{1}(x_{c})]\Big),
\end{align}
which can be evaluated numerically.

\begin{figure}[t!]
\vspace{-0.2cm}
\centering
\includegraphics[width=0.9\columnwidth]{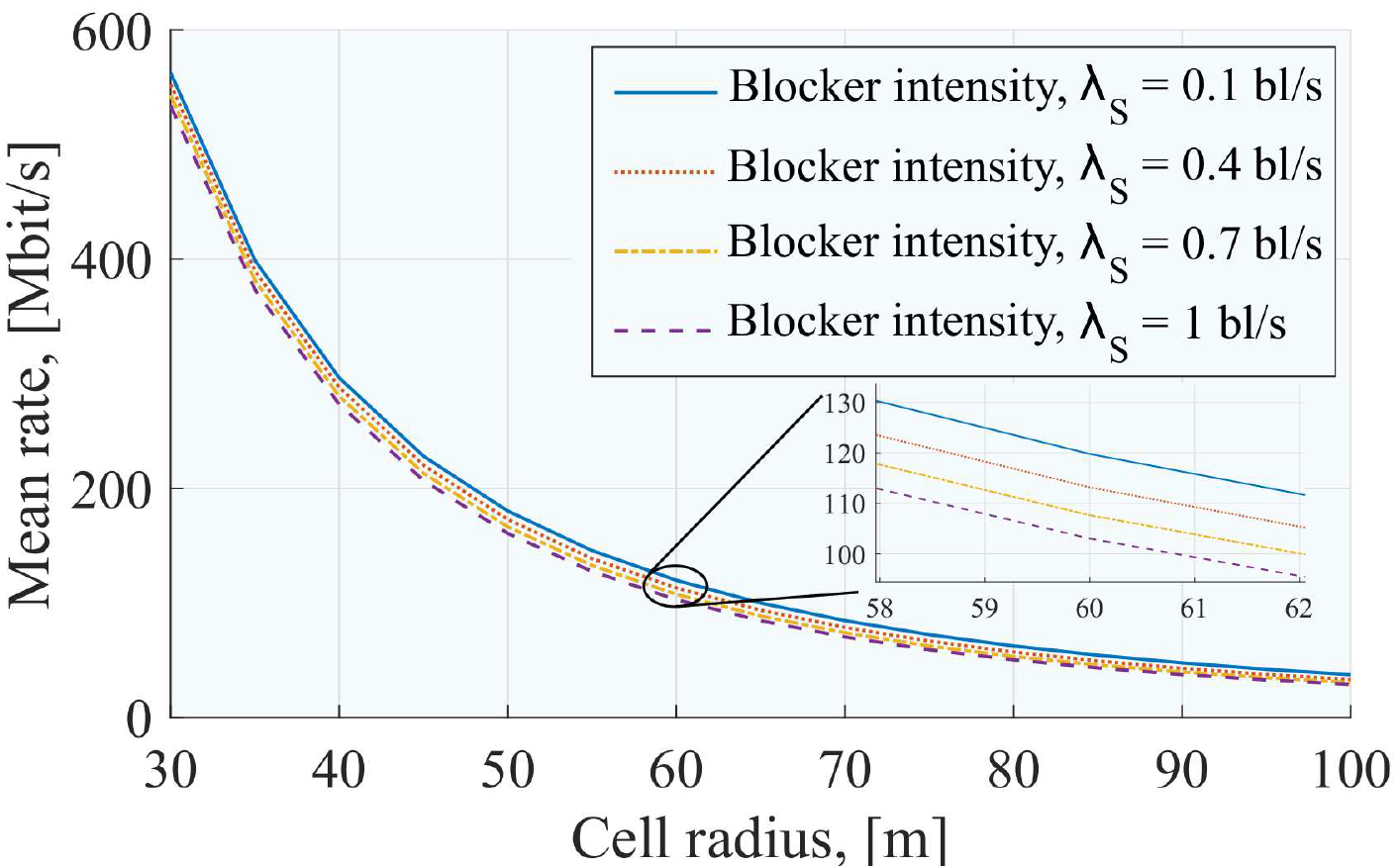}
\vspace{-0.2cm}
\caption{Mean data rate across the cell vs. cell radius, user, and blocker arrival intensities for the third scenario, S3.}
\vspace{-0.2cm}
\label{fig:rate_cell}
\end{figure}

The mean capacity made available to a user located at the cell edge $x_{c}$ is reported in Fig.~\ref{fig:rate_} for different user and blocker intensities. It is a function of the cell radius as well as the user and blocker intensities. The rest of the parameters are collected in Table~\ref{tab:sim_parameters}. In Fig.~\ref{fig:rate}, the density of the users is set as $\lambda_N = 0.01$~users/m$^2$, and the height of the mmWave Tx is assumed to be $h_T = 10$~m. As an example, these plots correspond to the third scenario of interest. As one may observe, the mean data rate decreases as the cell range and/or the arrival intensity of blockers grows. It should be noted that for equal density of users and cell radius, increase in the density of human blockers leads to the drop in the mean data rate. Provided with a particular target data rate, one may use Fig.~\ref{fig:rate_} to estimate the maximum cell radius for given blocker and user intensities, such that the chosen data rate is made available to all of the mmWave users.

In addition to the above, the analytical formulation for the mean data rate $R(x)$ of each user in the cell is given as
\begin{align}
\mathbb{E}[R]&=\sum_{N=1}^{\infty}\frac{(\lambda_N\pi{}x^{2})^N(e^{-\lambda_N\pi{}x^{2}})}{N!}c\frac{B}{N} \nonumber\\
&\times\int_{0}^{x_c}\Big(\mathbb{E}[T_l(x)]\log[1+s_{0}(x)] \nonumber\\
&+(1-\mathbb{E}[T_l(x)])\log[1+s_{1}(x)]\Big)f_{X_U}(x)dx,
\end{align}
where $f_{X_U}(x)=2x/x_{c}^{2}$ is the pdf of the distance $X_U$ between Tx and Rxs uniformly distributed in the cell of the radius $x_c$, while the fraction of time in non-blocked state is obtained by using (\ref{eqn:56_}) and (\ref{eqn:55_}) as
\begin{align}
\mathbb{E}[T_{l}(x)]=\frac{1}{\exp(\lambda \mathbb{E}[T])},
\end{align}
where $\mathbb{E}[T]$ is the mean residence time.

The mean rate of a randomly selected user in the cell is illustrated in Fig.~\ref{fig:rate_cell}. As one may observe, increasing the arrival intensity of blockers entering the unit area of the LoS blockage zone, $\lambda_S$, does not drastically affect the mean rate of a randomly selected user, as opposed to the mean rate at the cell edge.

The main points of the cell range analysis are summarized below:
\begin{itemize}
	\item{It is demonstrated that higher intensity of blockers decreases the mean rate of the user at the cell edge. By increasing the cell radius, the impact of the intensity of blockers becomes stronger, e.g., at the cell edge of 30 m the mean rate of a user is decreased by 1.7 times when the intensity grows by 10 times. However, when the cell radius is 100 m, the mean rate is 30 times lower.}
  \item{The mean rate of an arbitrarily chosen user decreases with the increased intensity of blockers. However, the effect of cell coverage on the mean rate is rather limited for all the considered intensities, e.g., 10 times higher intensity of blockers in the cell of radius 30~m decreases the mean rate by a factor of 1.05, whereas the same increase for 100~m radius cell decreases the mean rate by 1.3.}
\end{itemize}

\vspace{-0.2cm}
\subsection{System-Level Simulation Complexity}

Today, the performance of complex wireless systems is primarily assessed within large-scale system-level simulation (SLS, see e.g.,~\cite{Win2014}) environments. The proposed mathematical model can be efficiently utilized as part of an SLS tool to substitute for the need to explicitly model the blockage process. This may drastically improve the simulation run times, especially in highly crowded urban scenarios.

\begin{table}[!b]
\centering
\vspace{-0.2cm}
\caption{Absolute run time measurements in SLS evaluation, s.}
\begin{tabular}{|l|l|l|l|l|l|l|}
\hline
\multicolumn{2}{|c|}{\diagbox{$T_u$, ms}{$\lambda_I$, bl./s.}} & 0.1 & 0.3 & 0.5 & 0.7 & 1.0\\
\hline \hline
\multirow{2}{*}{100} & simulation & 0.101 & 0.215 & 0.532 & 0.860 & 0.928 \\
 \hhline{~------}
       & analysis & 0.250 & 0.272 & 0.243 & 0.256 & 0.292 \\
 \hline

\multirow{2}{*}{70} & simulation & 0.208 & 0.290 & 0.820 & 1.356 & 2.801 \\
 \hhline{~------}
       & analysis & 0.342 & 0.398 & 0.287 & 0.316 & 0.351 \\
 \hline

\multirow{2}{*}{50} & simulation & 0.581 & 1.012 & 1.91 & 3.282 & 5.982 \\
 \hhline{~------}
       & analysis & 0.681 & 0.538 & 0.369 & 0.694 & 0.499 \\
 \hline

\multirow{2}{*}{10} & simulation & 1.211 & 3.921 & 5.867 & 7.119 & 10.92\\
\hhline{~------}
      & analysis & 2.968 & 3.690 & 1.762 & 2.774 & 2.104 \\
 \hline

 \multirow{2}{*}{1} & simulation & 10.28 & 21.40 & 54.91 & 78.92 & 1.24e2\\
 \hhline{~------}
    & analysis & 22.64 & 18.18 & 23.98 & 19.09 & 22.11 \\
 \hline
\end{tabular}
\vspace{-0.3cm}
\label{tab:comp_complex}
\end{table}


In Table~\ref{tab:comp_complex}, the computation complexity measurements is reported in terms of the SLS run time as a function of the blocker arrival intensity and the \textit{environment update interval}, $T_U$ for the two cases: (i) direct simulation of the blockage process and (ii) application of the proposed model. From the SLS perspective, the environment update interval corresponds to how frequently the state of the users is monitored. In the latter case, each user has been associated with the pdf of the non-blocked and blocked intervals, thus implying that there is a need to update its state whenever the said interval expires. In the former direct modeling approach, at each environment update interval, one has to re-estimate the state of the users by employing the straightforward geometry considerations. Doing so significantly increases the computation complexity of the SLS evaluation, especially in dense environments. The experiments were conducted for the following parameters: the distance between the mmWave AP and the user is $r_{0}=10$~m and the blocker speed is $V=1$~m/s. The simulation time was set to $50$~s, while the hardware parameters were Intel Core i7-6700HQ CPU, 2.60~GHz (1 core run), and 32~GB RAM.


As it can be established by analyzing Table~\ref{tab:comp_complex}, the complexity of both simulation and analysis grows as the environment update interval decreases. Even though the simulation run time does not depend on the blocker density (nor on the distance between the mmWave AP and the Rx), the SLS modeling complexity increases with a higher number of blockers. This is because the computation complexity is associated with the need to characterize an intersection with every blocker to determine the current state of each user. For the SLS results reported in this work, the blockers are deployed at the edge of the modeled scenario and move across the street. From the simulation perspective, computation complexity grows linearly as blocker arrival intensity increases, i.e., the overall modeling complexity is $O(n)$. Although the use of sophisticated techniques, such as spatial hashing, may reduce the complexity down to $O(log(n))$ at the expense of more cumbersome implementation, the resulting complexity would still grow rapidly for higher user densities. In stark contrast, with the proposed analytical modeling, the complexity remains constant at $O(1)$. Finally, with the decreasing update interval $T_U$, both analytical and simulation complexity grow linearly ($O(n)$). However, it may not be as important, because the value of $T_U<$1~ms is seldom used in practical systems.

\vspace{-0.2cm}
\section{Conclusions}\label{sec:conclusions}


This work is aimed at a systematic characterization of the effects caused by the LoS blockage in cellular mmWave systems in presence of mobile blockers. To this end, three representative urban scenarios -- as discussed in the current 3GPP specifications -- were considered. The underlying process in the proposed mathematical approach was shown to be of alternating renewal nature, where non-blocked periods interchange with blocked intervals. The distribution of the non-blocked intervals was characterized by a simple memoryless exponential formulation, while the blocked periods were established to follow a general distribution. As example applications of the model, the height optimization of the mmWave AP, the mmWave cell range analysis, and the system-level modeling complexity reduction were considered.
 
Relying on the developed mathematical methodology, the impact of the LoS blockage is analyzed by establishing that the mean time in the blocked state is around 400-1000 ms for the typical input parameters, which amounts to a significant number of mmWave cellular superframes (around 20-50 according to \cite{nokia}). Moreover, a strong temporal correlation for the timescales of interest in mmWave systems (i.e., less than about a second) was demonstrated. The contributed temporal analysis could be useful for modeling human body blockage in the mmWave-specific system-level evaluation tools as well as when designing the mmWave-centric communication protocols.

\vspace{-0.2cm}
\begin{appendices}

\section{Residence Time in the LoS Blockage Zone}\label{appendix:appendix_A}

Here, the CDFs of the residence time in the LoS blockage zone is derived for the sidewalk 2 and the park/square/stadium scenarios (see Fig.~\ref{fig:analytic}).

\textit{Second scenario, S2.} Consider the sidewalk 2 scenario. Here, the main difference as compared to the sidewalk 1 scenario is in that the users tend to move closer to the central lane of the street. We model this effect by using the triangular distribution with the following CDF
\begin{align}\label{eqn:coord5}
F_Y(x)=
\begin{cases}
0, \qquad{}\qquad{}\qquad{}\qquad{}x \leq 0, \\
\displaystyle{\frac{x^2}{w_S c}}, \qquad{}\qquad{}\qquad{}\,\,\,\,0 < x \leq c, \\
\displaystyle{1 - \frac{(w_S - x)^2}{w_S (w_S - c)}},\,\,\,\,\,\,\,\,\,\, c < x \leq w_S, \\
1, \qquad{}\qquad{}\qquad{}\qquad{}x > w_S,
\end{cases}
\end{align}
where $c$ is the mode of the triangular distribution, which denotes the point with the highest probability density.

The distribution of distance, which is traveled by a blocker in the blockage zone, depends on the position of the LoS blockage zone with respect to the mode of the triangular distribution. The following five different cases are observed:
\begin{enumerate}
  \item{If $y_C \leq c$:}
\begin{align}\label{eqn:coord6}
F_{L}^{2,1}(x) =
\begin{cases}
0,\qquad{}\qquad{}x \leq 0, \\
\displaystyle{\frac{x \sin(2 \alpha)}{y_C - y_A}},\,\,\,0 < x \leq x_{min}, \\
1,\qquad{}\qquad{}x > x_{min}.
\end{cases}
\end{align}
  \item{If $y_C - y_{min} \leq c < y_C$, $y_{min} = x_{min} \cos(\alpha) \sin(\alpha)$:}
\begin{align}\label{eqn:coord7}
\hspace{-1.5em}F_{L}^{2,2}(x) =
\begin{cases}
0, x \leq 0, \\
\frac{\sin(2\alpha) (4 c (y_A + y_C) - w_S (x \sin(2\alpha) + 4(c+y_A)))}{4w_S (y_A^2 + c (c - 2 y_C)) + c (y_C^2 - y_A^2)}, \\
0 \leq x < \frac{y_C - c}{\sin(\alpha) \cos(\alpha)},\\
\frac{w_S (c - y_C)^2 + x \sin(2 \alpha) (c - w_S) (y_A + y_C)}{w_S (y_A^2 + c (c - 2 y_C)) + c (y_C^2 - y_A^2)}, \\
\frac{y_C - c}{\sin(\alpha) \cos(\alpha)} < x \leq x_{min}, \\
1, x > x_{min}.
\end{cases}
\end{align}
  \item{If $y_A + y_{min} \leq c < y_C - y_{min}$:}
\begin{align}\label{eqn:coord8}
F_{L}^{2,3}(x) =
\begin{cases}
0, x \leq 0, \\
\frac{x \sin(2 \alpha) (4 c (y_A + y_C) - w_S (x \sin(2 \alpha) + 4 (c + y_A)))}{4 (w_S (y_A^2 + c (c - 2 y_C)) + c (y_C^2 - y_A^2))}, \\
0 < x \leq x_{min}, \\
1, x > x_{min}.
\end{cases}
\end{align}
  \item{If $y_A \leq c < y_A + y_{min}$:}
\begin{align}\label{eqn:coord9}
F_{L}^{2,4}(x) =
\begin{cases}
0, x \leq 0, \\
\frac{x \sin(2 \alpha) (4 c (y_A + y_C) - w_S (x \sin(2 \alpha) + 4 (c + y_A)))}{4 (w_S (y_A^2 + c (c - 2 y_C)) + c (y_C^2 - y_A^2))}, \\
0 < x \leq \frac{c - y_A}{\sin(\alpha) \cos(\alpha)}, \\
\frac{w_S (c - y_A)^2 + x c \sin(2 \alpha) (y_A + y_C - 2 w_S)}{w_S (y_A^2 + c (c - 2 y_C)) + c (y_C^2 - y_A^2)}, \\
\frac{c - y_A}{\sin(\alpha) \cos(\alpha)} < x \leq x_{min}, \\
1, x > x_{min}.
\end{cases}
\end{align}
  \item{If $c \leq y_A$: }
\begin{align}\label{eqn:coord10}
F_{L}^{2,5}(x) =
\begin{cases}
0, x \leq 0, \\
\frac{x \sin(2 \alpha)}{y_C - y_A}, 0 < x \leq x_{min}, \\
1, x > x_{min}.
\end{cases}
\end{align}
\end{enumerate}

For the second scenario, the arrival intensity of the blockers that enter the LoS blockage zone is given by
\begin{align}\label{eqn:543}
\lambda = \lambda_I \Big(F_Y(y_{C})-F_Y(y_{A})\Big).
\end{align}

The above may be explained by the fact that the majority of blockers cross the width of the sidewalk in the middle by following the triangular distribution for the entry point. Therefore, the Poisson process in time has the arrival intensity of blockers emerging at the effective width per time unit, $w_E$, equal to $\lambda$ as derived in (\ref{eqn:543}).

\textit{Third scenario, S3.} Finally, for the park/stadium/square scenario, the CDF of distance walked by a blocker in the blockage zone is given by
\begin{align}\label{eqn:coord11}
\hspace{-0.5em}F_{L}^3(x) =
\begin{cases}
0,\,x \leq 0, \\
w_1 F_{L}^{3,1}(x) + w_2 F_{L}^{3,2}(x),\,\\
0 < x \leq \sqrt{d_{m}^{2} + r^2}, \\
1,\,x > \sqrt{d_{m}^{2} + r^2},
\end{cases}
\end{align}
where the weights $w_1$ and $w_2$ are the probability for a blocker to enter from the long side (AD or CB, see Fig.~\ref{fig:geometrical_scenario}) and to leave from another long side (AD or CB), and the probability for a blocker to enter from the short side (DC) and to leave from the long side (AD or CB), respectively, which are given by
\begin{align}
&w_1 = \frac{d_{m}^2 + 3 d_{m} r}{d_{m}^2 + 3 d_{m} r + 2 r^2},\nonumber\\
&w_2 = \frac{2 r^2}{d_{m}^2 + 3 d_{m} r + 2 r^2},
\end{align}
and the corresponding CDFs are
\begin{align}
F_{L}^{3,2}(x) =
\begin{cases}
0, \,x\leq d_{m}, \\
\frac{d_{m}^2 - x^2 + 2 r \sqrt{x^2 - d_{m}^2}}{r^2}, \,\\
d_{m} < x \leq \sqrt{d_{m}^2 + r^2}, \\
1, \,x > \sqrt{d_{m}^2 + r^2},
\end{cases}
\end{align}
and
\begin{align}\label{eqn:coord12}
&F_{L}^{3,1}(x) =
\begin{cases}
0, x \leq 0, \\
\frac{\pi x^2}{4 r d_{m}}, 0 < x \leq \min(r,d_{m}), \\
\frac{1}{2 r d_{m}} (\min(r,d_{m}) \sqrt{x^2 - \min(r,d_{m})^2} \\
+x^2 \arcsin(\frac{\min(r,d_{m})}{x})),\\
\min(r,d_{m}) < x \leq \max(r,d_{m}), \\
\frac{1}{2 r d_{m}} (\min(r,d_{m}) \sqrt{\max(r,d_{m})^2 - \min(r,d_{m})^2} \\
+d_{m} (\sqrt{x^2 - d_{m}^2} - \sqrt{\max(r,d_{m})^2 - d_{m}^2})  \\
+r (\sqrt{x^2 - r^2} - \sqrt{\max(r,d_{m})^2 - r^2}) \\
+\max(r,d_{m})^2 (\arccos(\frac{r}{\max(r,d_{m})})  \\
+\arcsin(\frac{\min(r,d_{m})}{\max(r,d_{m})}) - \arcsin(\frac{d_{m}}{\max(r,d_{m})}))  \\
+ x^2 (\arcsin(\frac{d_{m}}{x}) - \arccos(\frac{r}{x}))),\\
\max(r,d_{m}) < x \leq \sqrt{d_{m}^2 + r^2}, \\
1, x > \sqrt{d_{m}^2 + r^2}.
\end{cases}
\end{align}

The arrival intensity of the blockers that enter the zone, which affects the LoS for the third scenario, is then $\lambda = \lambda_I$.

\end{appendices}


\bibliographystyle{ieeetr}
\bibliography{mmW_Mobility_of_Blockers}
\begin{IEEEbiography}[{\includegraphics[width=1in,height=1.25in,clip,keepaspectratio]{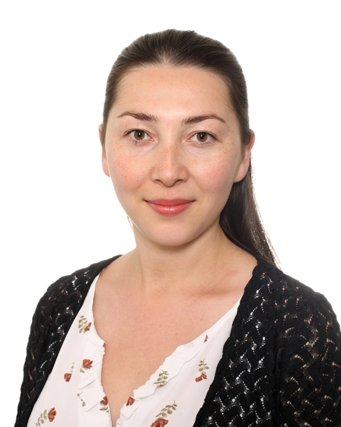}}]{Margarita Gapeyenko}
is a Ph.D. candidate at the Laboratory of Electronics and Communications Engineering at Tampere University of Technology, Finland. She earned her M.Sc. degree in Telecommunication Engineering from University of Vaasa, Finland, in 2014, and B.Sc. degree in Radio-Engineering, Electronics, and Telecommunications from Karaganda State Technical University, Kazakhstan, in 2012. Her research interests include mathematical analysis, performance evaluation, and optimization methods of future wireless networks, device-to-device communication, and 5G-grade heterogeneous networks. 
\end{IEEEbiography}

\begin{IEEEbiography}[{\includegraphics[width=1in,height=1.25in,clip,keepaspectratio]{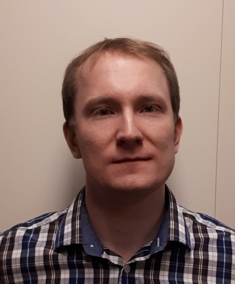}}]{Andrey Samuylov}
received the Ms.C. in Applied Mathematics and Cand.Sc. in Physics and Mathematics from the RUDN University, Russia, in 2012 and 2015, respectively. Since 2015 he is working at Tampere University of Technology as a researcher, focusing on mathematical performance analysis of various 5G wireless networks technologies. His research interests include P2P networks performance analysis, performance evaluation of wireless networks with enabled D2D communications, and mmWave band communications. 
\end{IEEEbiography}

\begin{IEEEbiography}[{\includegraphics[width=1in,height=1.25in,clip,keepaspectratio]{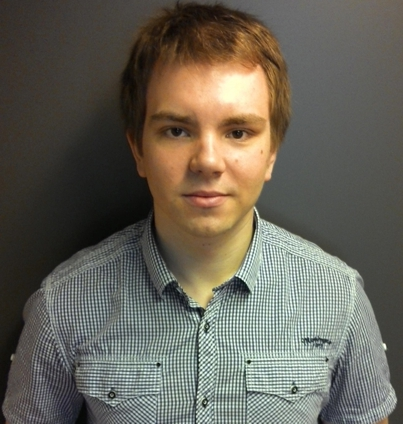}}]{Mikhail Gerasimenko}
is a Researcher at Tampere University of Technology in the Laboratory of Electronics and Communications Engineering. Mikhail received Specialist degree from Saint-Petersburg University of Telecommunications in 2011. In 2013, he obtained Master of Science degree from Tampere University of Technology. Mikhail started his academic career in 2011 and since then he appeared as (co-)author of multiple scientific journal and conference publications, as well as several patents. He also acted as reviewer and participated in educational activities. His main subjects of interest are wireless communications, machine-type communications, and heterogeneous networks. 
\end{IEEEbiography}

\begin{IEEEbiography}[{\includegraphics[width=1in,height=1.25in,clip,keepaspectratio]{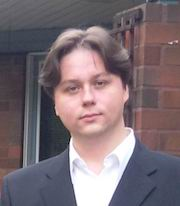}}]{Dmitri Moltchanov}
is a Senior Research Scientist in the Laboratory of Electronics and Communications Engineering, Tampere University of Technology, Finland. He received his M.Sc. and Cand.Sc. degrees from Saint-Petersburg State University of Telecommunications, Russia, in 2000 and 2002, respectively, and Ph.D. degree from Tampere University of Technology in 2006. His research interests include performance evaluation and optimization issues of wired and wireless IP networks, Internet traffic dynamics, quality of user experience of real-time applications, and traffic localization P2P networks. Dmitri Moltchanov serves as TPC member in a number of international conferences. He authored more than 80 publications.  
\end{IEEEbiography}

\begin{IEEEbiography}[{\includegraphics[width=1in,height=1.25in,clip,keepaspectratio]{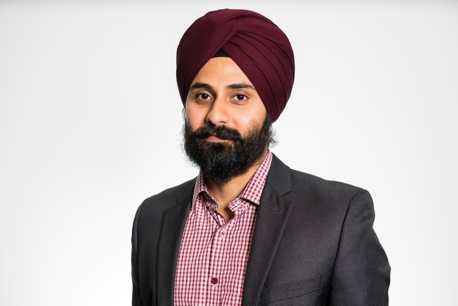}}]{Sarabjot Singh}
('SM 09, M' 15) is a Principal Engineer at Uhana Inc. CA.  He received the B. Tech. from IIT, India, and  the M.S.E and Ph.D. in EE from UT Austin. His past affiliations include Intel,  Nokia Technologies, Bell Labs, and Qualcomm Inc, where he worked on protocol and algorithm design for next generation of cellular and WiFi networks. Dr. Singh is interested in the system and architecture design of wireless networks leveraging theoretical and applied tools. He is a co-author of more than 40 patent applications, and multiple conference and journal papers. He was the recipient of the President of India Gold Medal in 2010, the ICC Best Paper Award in 2013, and recognized for being a prolific inventor at Intel Corp. 
\end{IEEEbiography}

\begin{IEEEbiography}[{\includegraphics[width=1in,height=1.25in,clip,keepaspectratio]{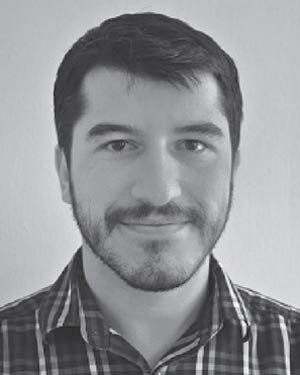}}]{Mustafa Riza Akdeniz}
(S'09) received the B.S. degree in electrical and electronics engineering from Bogazici University, Istanbul, Turkey, in 2010 and the Ph.D. degree in electrical and computer engineering at New York University Tandon School of Engineering, Brooklyn, NY in 2016. He is working as a research scientist for Intel Labs in Santa Clara, CA. His research interests include wireless communications, wireless channel modeling, and information theory. 
\end{IEEEbiography}

\begin{IEEEbiography}[{\includegraphics[width=1in,height=1.25in,clip,keepaspectratio]{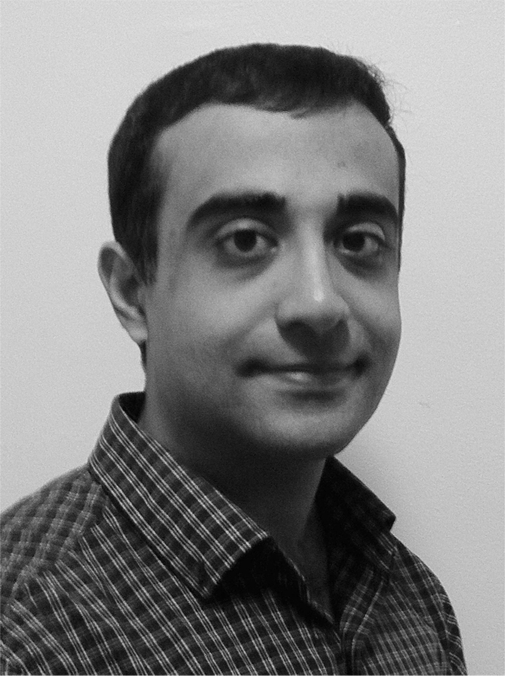}}]{Ehsan Aryafar}
is an Assistant Professor of Computer Science at Portland State University. Prior to that and from 2013 to 2017, he was a Research Scientist at Intel Labs in Santa Clara, CA. He received the B.S. degree in Electrical Engineering from Sharif University of Technology, Iran, in 2005, and the M.S. and Ph.D. degrees in Electrical and Computer Engineering from Rice University, Houston, Texas, in 2007 and 2011, respectively. From 2011 to 2013, he was a Post-Doctoral Research Associate in the Department of Electrical Engineering at Princeton University. His research interests are in the areas of wireless networks and networked systems, and span both algorithm design as well as system prototyping. He has more than 30 issued and pending patents in the areas of mobile and wireless systems. 
\end{IEEEbiography}

\begin{IEEEbiography}[{\includegraphics[width=1in,height=1.25in,clip,keepaspectratio]{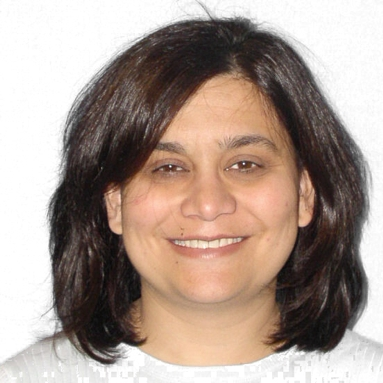}}]{Nageen Himayat}
is a Principal Engineer with Intel Labs, where she leads a team  conducting research on several aspects of next generation (5G/5G+) of mobile broadband systems. Her research contributions span areas such as multi-radio heterogeneous networks, mm-wave communication, energy-efficient designs, cross layer radio resource management, multi-antenna, and non-linear signal processing techniques. She has authored over 250 technical publications, contributing to several IEEE peer-reviewed publications, 3GPP/IEEE standards, as well as holds numerous patents. Prior to Intel, Dr. Himayat was with Lucent Technologies and General Instrument Corp, where she developed standards and systems for both wireless and wire-line broadband access networks. Dr. Himayat obtained her B.S.E.E degree from Rice University, and her Ph.D. degree from the University of Pennsylvania. She also holds an MBA degree from the Haas School of Business at University of California, Berkeley.  
\end{IEEEbiography}

\begin{IEEEbiography}[{\includegraphics[width=1in,height=1.25in,clip,keepaspectratio]{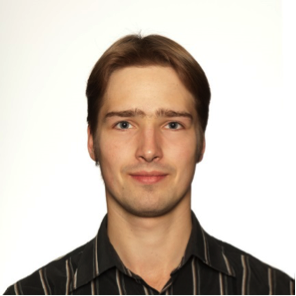}}]{Sergey Andreev}
is a Senior Research Scientist in the Laboratory of Electronics and Communications Engineering at Tampere University of Technology, Finland. He received the Specialist degree (2006) and the Cand.Sc. degree (2009) both from St. Petersburg State University of Aerospace Instrumentation, St. Petersburg, Russia, as well as the Ph.D. degree (2012) from Tampere University of Technology. Sergey (co-)authored more than 100 published research works on wireless communications, energy efficiency, heterogeneous networking, cooperative communications, and machine-to-machine applications.  
\end{IEEEbiography}

\begin{IEEEbiography}[{\includegraphics[width=1in,height=1.25in,clip,keepaspectratio]{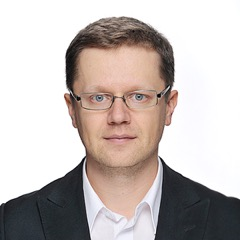}}]{Yevgeni Koucheryavy}
is a Full Professor in the Laboratory of Electronics and Communications Engineering of Tampere University of Technology (TUT), Finland. He received his Ph.D. degree (2004) from TUT. He is the author of numerous publications in the field of advanced wired and wireless networking and communications. His current research interests include various aspects in heterogeneous wireless communication networks and systems, the Internet of Things and its standardization, as well as nanocommunications. He is Associate Technical Editor of IEEE Communications Magazine and Editor of IEEE Communications Surveys and Tutorials.
\end{IEEEbiography}
\end{document}